\documentclass[10pt,twocolumn,twoside]{IEEEtran}
\usepackage{graphics} % for pdf, bitmapped graphics files
\usepackage{epsfig} % for postscript graphics files
\usepackage{amsmath} % assumes amsmath package installed
\usepackage{amssymb}  % assumes amsmath package installed
\usepackage{amsthm}
\usepackage{cite}
\usepackage{color}
\usepackage{enumerate}
\usepackage{graphicx}
\usepackage{mathtools}
\usepackage{cancel}
\usepackage{url}
\usepackage{mathdots}
\usepackage{MnSymbol}
\usepackage{pdftexcmds}
\usepackage[font=footnotesize,labelfont=bf]{caption}
\usepackage{subfigure}

\newcommand{\until}[1]{\{1,\dots,#1\}}
\newcommand{\map}[3]{#1:#2 \rightarrow #3}

\newcommand{\oprocendsymbol}{\hbox{$\bullet$}}
\newcommand{\oprocend}{\relax\ifmmode\else\unskip\hfill\fi\oprocendsymbol}

\newcommand{\longthmtitle}[1]{\mbox{}\textup{\textbf{(#1):}}}

\newcommand*{\SetSuchThat}[1][]{} % reserve the name
\newcommand*{\MvertSets}{%
    \renewcommand*\SetSuchThat[1][]{%
        \mathclose{}%
        \nonscript\;##1\vert\penalty\relpenalty\nonscript\;%
        \mathopen{}%
    }%
}
\MvertSets % default
\DeclarePairedDelimiterX \Set [2] {\lbrace}{\rbrace}
    {\,#1\SetSuchThat[\delimsize]#2\,}

\usepackage{tikz,calc,bbding}
\usetikzlibrary{shapes, arrows, decorations, decorations.pathreplacing,
decorations.pathmorphing, decorations.markings, fit, matrix}
\tikzset{->-/.style={decoration={
  markings,
  mark=at position #1 with {\arrow{>}}},postaction={decorate}}}

\newcommand{\R}{\mathbb{R}}

\newcommand{\Rplus}{\mathbb{R}_{> 0}}
\newcommand{\Rpluseq}{\mathbb{R}_{\geq 0}}

\newcommand{\A}{\mathcal{A}}

\newcommand{\norm}[1]{\left\Vert #1 \right\Vert}
\newcommand{\x}{\mathbf{x}}
\newcommand{\m}{\mathbf{m}}
\newcommand{\W}{\mathbf{W}}
\renewcommand{\c}{\mathbf{c}}
\newcommand{\B}{\mathbf{B}}
\newcommand{\K}{\mathbf{K}}
\newcommand{\eye}{\mathbf{I}}
\newcommand{\zero}{\mathbf{0}}
\renewcommand{\u}{\mathbf{u}}
\newcommand{\f}{\mathbf{f}}
\newcommand{\G}{\mathbf{G}}
\newcommand{\g}{\mathbf{g}}

\renewcommand{\d}{\mathbf{d}}
\newcommand{\w}{\mathbf{w}}
\newcommand{\F}{\mathbf{F}}
\newcommand{\all}{\mathrm{all}}

\newtheorem{theorem}{Theorem}[section]

\newtheorem{lemma}[theorem]{Lemma}

\newtheorem{corollary}[theorem]{Corollary}

\theoremstyle{definition}
\newtheorem{example}[theorem]{Example}
\newtheorem{remark}[theorem]{Remark}

\allowdisplaybreaks

\title{\LARGE \bf Selective Inhibition and Recruitment in
  Linear-Threshold Thalamocortical Networks\thanks{This work was
    supported by NSF Award CMMI-1826065.}}
\author{Michael McCreesh \qquad Jorge Cort\'es% <-this % stops a space
  \thanks{M. McCreesh and J. Cort\'es are with Department of
    Mechanical and Aerospace Engineering, UC San Diego,
    \{mmccreesh,cortes\}@ucsd.edu}}

\begin{document}
\maketitle

\begin{abstract}
  Neuroscientific evidence shows that for most brain networks all
  pathways between cortical regions either pass through the thalamus
  or a transthalamic parallel route exists for any direct
  corticocortical connection.  This paper seeks to formally study the
  dynamical behavior of the resulting thalamocortical brain networks
  with a view to characterizing the inhibitory role played by the
  thalamus and its benefits.  We employ a linear-threshold mesoscale
  model for individual brain subnetworks and study both hierarchical
  and star-connected thalamocortical networks. Using tools from
  singular perturbation theory and switched systems, we show that
  selective inhibition and recruitment can be achieved in such
  networks through a combination of feedback and feedforward control.
  Various simulations throughout the exposition illustrate the
  benefits resulting from the presence of the thalamus regarding
  failsafe mechanisms, required control magnitude, and network
  performance.
\end{abstract}

\vspace*{-1ex}
\section{Introduction}
The brain is a complex network composed of billions of individual
neurons, with their interconnections forming subnetworks that perform
a myriad of different functions. Communication of information between
regions and its subsequent processing is one such function.  Brain
regions, such as the neocortex for example, have a hierarchical
structure in which different cognitive levels operate on distinct
timescales. Within this hierarchy, information travels from faster
lower-level sensory brain regions to slower higher-level cognitive
brain regions (\emph{bottom-up} communication). Upon processing in the
higher-level regions, information regarding decisions made by these
regions is passed back down the hierarchy to perform some task
(\emph{top-down} communication). In this process, certain regions are
selectively recruited to perform the given task, while other areas are
selectively inhibited to ignore other inputs into the brain network.

Such hierarchies are not restricted to the neocortex, and neither
top-down and bottom-up communications occur entirely inside the
neocortex. In fact, most, if not all, direct corticocortical
communications have a parallel transthalamic pathway upon which the
information is transmitted and modulated~\cite{SMS:12}. Our goal here
is to understand the role of transthalamic communication in enabling
selective attention, with a view to characterizing its benefits. We
seek to provide a dynamical explanation of this phenomena and validate
the hypothesis that selective inhibition and recruitment are feasible
in thalamocortical networks via feedback and feedforward mechanisms.

\emph{Literature Review:}
The hierarchical organization of the brain has been known for
decades~\cite{NT:50,ARL:70} and has been extensively studied from
different
viewpoints~\cite{SJK-JD-KJF:08,JDM-AB-DJF-RR-JDW-XC-CP-TP-HS-DL-XW:14,DJF-DCEV:91,UH-JC-CJH:15,MIR-IT-PV:15}.
The role of the communication between the thalamus and cortical
regions in a thalamocortical hierarchy is a more recently studied
problem. Historically, the thalamus has been viewed as a relay of
sensory signals to the cortex. However, in recent literature, it has
been shown to also play a further role in cognitive
processes~\cite{KH-MAB-WBL-MD:17}. In
particular,~\cite{SMS:12,RDD-AB:17} show that the thalamus transfers
both sensory signals to the cortex using first-order relays, but that
for most direct corticocortical connection, there exists a parallel
transthalamic path made of higher-order thalamic relays. Further
evidence is given that these paths operate using feedforward
inhibitory control to communicate information from thalamic to
cortical
areas~\cite{MMH-LA-16,JMA-HAS:15,SMS-RWG:06,LG-SPJ-DEF-MC-MS:05,SJC-TJL-BWC:07}. The
works~\cite{JAH-SM-KEH-JDW-HC-AB-PB-SC-LC-AC:19,MMH-LA-16} show that
depending on the purpose (e.g., visual, auditory, somatomotor) of the
hierarchical network, the thalamus connects to the hierarchy at
different levels.  In general, little theoretical understanding is
available about the network properties of thalamocortical structures
and their role in the hierarchical nature of the brain. To address
this gap, here we employ linear-threshold
dynamics~\cite{RC-KK-MG-HK-XW:15,KM-AD-VI-CC:16} as a mesoscale model
for the behavior of neuronal populations and build on the hierarchical
selective recruitment framework introduced
in~\cite{EN-JC:21-tacI,EN-JC:21-tacII} for corticocortical
networks. The latter provides us with a baseline to compare against
when analyzing the properties and performance of thalamocortical
networks.  We also rely on results from switched piecewise and affine
systems~\cite{DL:03,MKJJ:03} and singular perturbation
theory~\cite{ANT:52,PVK-HKK:99,VV:97}.

\emph{Statement of Contributions:} We deal with thalamocortical brain
networks where each brain region is governed by a linear-threshold
rate dynamics. Given our focus on selective attention, the neuronal
populations in each region are divided into task-relevant and
task-irrelevent nodes.  Inspired by the types of pathways in thalamic
circuitry, we consider two interconnections topologies, multilayer
hierarchical networks and star-connected networks.  Our first
contribution provides an analysis of selective inhibition and
recruitment in hierarchical thalamocortical networks.  We establish
that the equilibrium maps of individual layers can be described as
piecewise-affine maps. Using singular perturbation theory for
non-smooth differential equations, we provide conditions for the
existence of feedback-feedforward control laws that achieve selective
inhibition and recruitment of the desired nodes.  Our second
contribution deals with star-connected thalamocortical networks, both
with and without timescale separation between regions.  For the latter
class, we build on a generalization of stability results on slowly
varying nonlinear systems to the case of exponential stability to
provide conditions for the existence of a feedback-feedforward
controller providing selective inhibition and recruitment. We achieve
analogous results for the case of star-connected networks with
timescale separation using again singular perturbation theory.
Examples illustrate the beneficial role played by the thalamus in
these networks with respect to metrics such as failsafe mechanisms,
control magnitude, and network performance.

\section{Neuroscientific Background}

Here we provide a summary of the neuroscientific background behind the
modeling assumptions adopted in the paper. We focus on observations
about brain organization, information transmission among brain
regions, and the role that the thalamus is believed to play.

\emph{Task-Relevant and Task-Irrelevant Neuron Populations:} in the
nervous system, stimuli are represented using series of electric
spikes generated by neurons that travel down nerve
fibres~\cite{PD-LFA:01}. A given stimulus is defined by a
characteristic pattern of spikes traveling between neurons, which can
also be represented as the firing rate of the neurons over time. In
such a representation, some neurons generate spikes (have a non-zero
firing rate) during the transmission of the stimuli, known as being
excited, while other neurons do not (have a firing rate of zero),
which is referred to as being inhibited\footnote{In network models,
  individual nodes are frequently considered to be populations of
  neurons with similar firing rates. Our discussion is not impacted by
  considering individual neurons or populations.}.  We refer to the
subset of neuron populations that are excited during the transmission
of a stimuli as the `task-relevant' nodes and the remaining
populations as the `task-irrelevant' nodes.

\emph{Information Pathways:} in the brain, there exist information
pathways between different spatial regions allowing for the
transmission of stimuli between processing areas. The transmission of
information can be seen as the activity in one region (that is, the
firing rates of the neuron populations defined by the stimuli) driving
the activity in the following region by exciting the task-relevant
nodes and inhibiting the task-irrelevant nodes to propagate the
stimuli (and any processing of it) through the pathway. This enables
the brain to generate appropriate responses to the stimuli by
propagating a response through the information pathway. This response,
with its own set of task relevant/irrelevant nodes, selectively
\emph{recruits} (excites) the task-relevant nodes, and selectively
\emph{inhibits} the task-irrelevant nodes.  Information pathways
between brain regions form both spatial and temporal hierarchies,
allowing for different levels of processing occurring in different
regions~\cite{JC-UH-CJH:15}. The temporal timescales are directly
related to the complexity of processing occurring in a given
region. For a low-level sensory area, where the inputs are brief and
are processed quickly, the timescales are fast.  Meanwhile, further up
the information pathway, in regions such as the prefrontal cortex,
higher-level cognitive processes use more complex inputs from earlier
in the pathway by integrating them over time, resulting in slower
timescales~\cite{RC-KK-MG-HK-XW:15}. These pathways have been studied
in strictly cortical
networks~\cite{RC-KK-MG-HK-XW:15,EN-JC:21-tacI,EN-JC:21-tacII}.
However, elements of such networks also have multiple connections to
the thalamus while maintaining a temporal hierarchy~\cite{SMS:12},
leading to the interest here in the role of the thalamus.

\begin{figure}[tbh]
  \centering
  \subfigure[hierarchical thalamocortical network]{
    \begin{tikzpicture}
      \node[circle, draw, line width=0.6pt, inner sep=-0.4pt] (1)
      {\includegraphics[width=24pt]{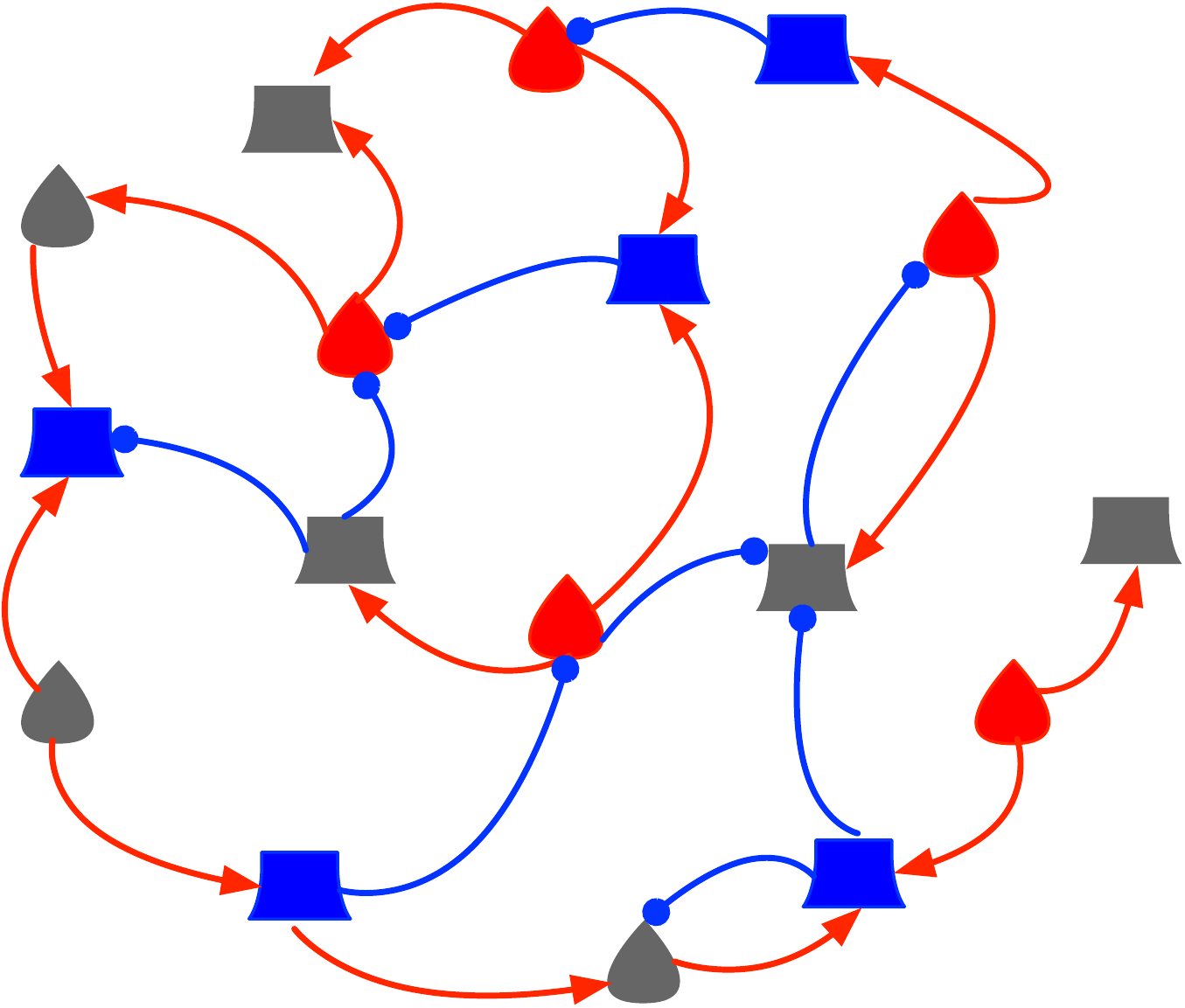}};
      % \draw[-, line width=0.3pt] (1.north) to (1.south);
      \node[above of=1, yshift=30pt, circle, draw, line
      width=0.6pt, inner sep=-0.4pt] (2)
      {\includegraphics[width=24pt]{EI_network}};
      % \draw[-, line width=0.3pt] (2.north) to (2.south);
      \node[below of=1, yshift=-30pt, circle, draw, line
      width=0.6pt, inner sep=-0.4pt] (3)
      {\includegraphics[width=24pt]{EI_network}};
      % \draw[-, line width=0.3pt] (3.north) to (3.south);
      \draw[latex-latex, line width=0.3pt, shorten <=1pt, shorten
      >=1pt] (1.90) to (2.270);
      \draw[latex-latex, line width=0.3pt, shorten <=1pt, shorten
      >=1pt] (1.270) to (3.90);
      \node[below of=3, yshift=5pt] {$\vdots$};
      \node[above of=2, yshift=-5pt] {$\vdots$};
      \node[left of=2, xshift=-20pt, yshift=10pt, scale=0.8] (i-1)
      {Subnetwork $i - 1$};
      \node[left of=1, xshift=-20pt, yshift=10pt, scale=0.8] (i) {Subnetwork $i$};
      \node[left of=3, xshift=-20pt, yshift=10pt, scale=0.8] (i+1) {Subnetwork $i +
        1$};
      \node[right of=1, xshift=30pt, circle, draw, line
      width=0.6pt, inner sep=-0.4pt] (i-big)
      {\includegraphics[width=24pt]{EI_network}};
      \draw[latex-latex,shorten <= 1pt,shorten >= 1pt] (1) to (i-big);
      \draw[latex-latex,shorten <= 1pt,shorten >= 1pt] (2) to (i-big);
      \draw[latex-latex,shorten <= 1pt,shorten >= 1pt] (3) to (i-big);
      \node[above of=i-big, xshift=0pt, yshift=-1pt, scale=0.8] (i1)
      {\parbox{40pt}{\centering  Thalamus}};
    \end{tikzpicture}
  }
  \qquad
  \subfigure[star-connected
  thalamocortical network]{\begin{tikzpicture} \node[circle, draw, line
      width=0.6pt, inner sep=-0.4pt] (1)
      {\includegraphics[width=24pt]{EI_network}};
      % \draw[-, line width=0.3pt] (1.north) to (1.south);
      \node[above of=1, xshift = 25pt, yshift = -5pt, scale=0.8]{Thalamus};
      \node[above of=1, yshift=30pt, circle, draw, line
      width=0.6pt, inner sep=-0.4pt] (2)
      {\includegraphics[width=24pt]{EI_network}};
      % \draw[-, line width=0.3pt] (2.north) to (2.south);
      \node[above of=2, xshift = 0pt, yshift = -5pt, scale=0.8]{Subnetwork $1$};
      \draw[latex-latex, line width=0.3pt, shorten <=1pt, shorten
      >=1pt] (2.270) to (1.90);
      \node[left of=1, xshift=-30pt, circle, draw, line
      width=0.6pt, inner sep=-0.4pt] (3)
      {\includegraphics[width=24pt]{EI_network}};
      % \draw[-, line width=0.3pt] (3.north) to (3.south);
      \node[below of=3, xshift = 0pt, yshift = 5pt, scale=0.8]{Subnetwork $4$};
      \draw[latex-latex, line width=0.3pt, shorten <=1pt, shorten
      >=1pt] (3.0) to (1.180);
      
      \node[right of=1, xshift=30pt, circle, draw, line
      width=0.6pt, inner sep=-0.4pt] (4)
      {\includegraphics[width=24pt]{EI_network}};
      % \draw[-, line width=0.3pt] (4.north) to (4.south); 
      \draw[latex-latex, line width=0.3pt, shorten <=1pt, shorten
      >=1pt] (1.0) to (4.180);
      \node[below of=4, xshift = 0pt, yshift = 5pt, scale=0.8]{Subnetwork $2$};  
      
      \node[below of=1, yshift=-30pt, circle, draw, line
      width=0.6pt, inner sep=-0.4pt] (5)
      {\includegraphics[width=24pt]{EI_network}};
      % \draw[-, line width=0.3pt] (5.north) to (5.south); 
      \node[below of=5, xshift = 0pt, yshift = 0pt, scale=0.8]{Subnetwork $3$};
      \draw[latex-latex, line width=0.3pt, shorten <=1pt, shorten
      >=1pt] (5.90) to (1.270);
    \end{tikzpicture}
}
\caption{Topologies for thalamocortical networks considered in the
  paper. In (a), each layer is connected directly to the thalamus, as
  well as the layers directly above and below it.  In (b), each layer
  is connected to the thalamus, but no direct connections exist
  between cortical regions.  In both plots, task-relevant excitatory
  and inhibitory nodes are depicted in red and blue, resp., and
  task-irrelevant nodes are depicted in
  grey. }\label{fig:thalamocortical-networks}
\vspace*{-2.5ex}
\end{figure}

\emph{Pathways between the Thalamus and Cortical Regions:} in studying
thalamocortical networks, we consider two topologies of interest, cf.
Fig.~\ref{fig:thalamocortical-networks}. Both share the common trait
that each cortical layer is connected to the thalamus layer: however,
connections between cortical layers differ in each case.  These
topologies are inspired by the fact that pathways in thalamic
circuitry can be identified into two classes~\cite{SMS-RWG:11}. The
first represents the role of the thalamus as a modulator of
information being passed between cortical regions along a
transthalamic route parallel to the existing corticocortical
information pathways (higher-order relay)~\cite{SMS:12}. Since these
cortical regions then form a temporal hierarchy, this leads to the
hierarchical thalamocortical network shown in
Fig.~\ref{fig:thalamocortical-networks}(a). The second class
represents the thalamus as the main route for the transfer of
information between two (or more) brain areas. In this case, the thalamus is
relaying an input to the cortical regions (first-order
relay)~\cite{SMS:12}, which gives rise to the star-connected network
shown in Fig.~\ref{fig:thalamocortical-networks}(b).
The cortical regions to which information is being relayed can be
parts of separate temporal hierarchies, however, as the
regions in the network do not form a temporal hierarchy themselves,
the timescales of the subnetworks are not directly related.

\emph{Role of the Thalamus:} The neuroscientific
literature~\cite{SMS-RWG:06} discusses a number of potential roles for
the thalamus.  The thalamus can significantly increase the information
contained in signals both being transmitted to and between cortical
regions~\cite{SMS:12}.  In hierarchical networks within the cortex,
transthalamic pathways allow for layers near the top of the hierarchy
to directly receive the outputs from the lower levels, in addition to
the more processed inputs they receive from other cortical
regions~\cite{SMS-RWG:06}. Thalamic signals into the higher cortical
regions can also allow for receiving further details of motor signals,
such as distinguishing between self-induced stimuli (such as those
generated by eye movements) and stimuli from the external
environment~\cite{SMS:12}. The thalamus also has a role in controlling
recurrent cortical dynamics, as the cortical networks are not able to
self-sustain activity~\cite{KR-ADL-MS:15,JMA-HAS:15}. Other roles
include contributions to learning, memory, decision-making and
inhibitory
control~\cite{ASM-SMS-MAS-RGM-RPV-YC:14,FA-VF-ARM-EJK-EC-WM:18}.

\section{Problem Setup}\label{sec:problem-setup}

We start\footnote{We let $\R$, $\Rpluseq$, denote the reals and
  nonnegative reals, resp. Vectors and matrices are identified
  by bold-faced letters. For vectors (matrices)
  $\x,\mathbf{y} \in \R^n$ (resp. $\R^{n \times m}$),
  $\x \leq \mathbf{y}$ is the component-wise comparison (analogously
  with $<,>,\geq)$.
  % A positive
  % definite matrix $\W \in \R^{n\times n}$ is denoted $\W \succ 0$
  % (analogously with $\succeq,\preceq,\prec$).
   The identity matrix of
  dimension $n$ is~$\eye_n$. $\zero_n$ and $\zero_{n \times m}$ denote
  the $n$-vector of zeros and the $n \times m$ matrix of zeros,
  resp. For a matrix $\W \in \R^{n \times n}$, we denote its
  element-wise absolute value, spectral radius and
  induced $2$-norm by $|\W| ,\rho(\W)$ and $\norm{\W}$,
  resp. Similarly, we let $\norm{\x}$ denote the $2$-norm of a
  vector $\x \in \R^n$. 
  % For a set $A \subset \R^n$ we denote the
  % shortest distance from $\x \in \R^n$ to the set by
  % $|\x|_A = \min_{\mathbf{y} \in A} \norm{\x-\mathbf{y}}$.
   For $x \in \R$ and $m \in \Rplus$, $[x]_0^m$ denotes
  $\min\{\max\{x,0\},m\}$. For $\x\in \R^n$, $\m \in \Rplus^n$, this
  operation is done component-wise as
  $[\x]_\zero^\m = [ [x_1]_0^{m_1}, \dots, [x_n]_0^{m_n} ]$. For a
  $2 \times 2$-partitioned block matrix
$
  \W = 
  \begin{bmatrix}
    {\bf W}^{00} & {\bf W}^{01}
    \\
    {\bf W}^{10} & {\bf W}^{11}
  \end{bmatrix}$, we use the notation
  $\W^{\ell,\all} = \begin{bmatrix}{\bf W}^{\ell 0} & {\bf W}^{\ell
      1} \end{bmatrix}$ and
  $\W^{\all,\ell} = \begin{bmatrix}(\W^{0\ell})^\top &
    (\W^{1\ell})^\top \end{bmatrix}^\top$ for $\ell \in \{0,1\}$.
}
 by providing details on the dynamic modeling of the
thalamocortical network layers, then describe the effect that the
interconnection topology has on the input to each layer, and finally
formalize the problem under consideration.

\subsection{Network Modeling}\label{sec:modeling}

We consider a thalamocortical network $\mathcal{N}$ composed of $N$
cortical layers and the thalamus.  We use linear-threshold rate
dynamics to model the evolution of each region in the network. These
dynamics provide a mesoscale model of the evolution of the average
firing rate of populations of neurons, rather than individual spike
trains, by looking at the electrical currents flowing through synaptic
connections, see~\cite{PD-LFA:01}.  The dynamics of cortical layer
$\mathcal{N}_i$ composed of $n_i$ nodes are
\begin{align}\label{eq:lin_threshold_dynamics}
  \tau_i \dot{\x}_i = -\x_i + [\W_i\x_i + \mathbf{d}_i(t)]_\zero^{\m_i} \qquad
  \zero \leq \x_i(0) \leq \m_i, 
\end{align}
where $\x_i \in \Rpluseq^{n_i}$ represents the state of the nodes
within the layer, and each component of $\x_i$ represents a population
of neurons with similar firing rate. The matrix $\W_i \in \R^{n_i
  \times n_i}$ is the synaptic connectivity between neuron populations
within the layer and $\d_i(t) \in \R^{n_i}$ encapsulates the input
into the model,
\begin{align}\label{eq:lin_threshold_input_expansion}
  \d_i(t) = \w_i(t) + \B_i\u_i(t) + \c_i.
\end{align}
Here $\w_i(t)$ models the interconnections between layers, $\u_i(t)$
is the control used for the $n_i$ nodes in $\mathcal{N}_i$, and $\c_i$
includes any unmodeled background activity or external inputs. Finally
$\tau_i \in \Rpluseq^{n_i}$ is the timescale of the dynamics.

In our study, the control $\u_i$ inhibits $r_i \leq n_i$
task-irrelevant nodes (drives their state to zero), while the
remaining terms in~\eqref{eq:lin_threshold_input_expansion} recruit
the remaining $n_i-r_i$ nodes and determine the desired equilibrium or
trajectory.  To distinguish between task-relevant and task-irrelevant
nodes, we use the following partition of network variables in each
layer $\mathcal{N}_i$,
\begin{subequations}\label{eq:relevant-decomposition}
  \begin{align}
    \x_i = \begin{bmatrix}
      \x_i^0 \\
      \x_i^1
    \end{bmatrix} \qquad \W_{i} &= \begin{bmatrix}
      \W_{i}^{00} & \W_{i}^{01} \\
      \W_{i}^{10} & \W_{i}^{11}
    \end{bmatrix}
    \label{eq:relevant1}
    \\
    \B_i = \begin{bmatrix}
      \B_i^0 \\
      \zero
    \end{bmatrix} \qquad \c_i &= \begin{bmatrix}
      \c_i^0 \\ \c_i^1
    \end{bmatrix} \qquad \m_i = \begin{bmatrix}
      \m_i^0 \\ \m_i^1
    \end{bmatrix},
    \label{eq:relevant2}
  \end{align}
\end{subequations}
where $\x_i^0 \in \Rpluseq^{r_i}$ represent the task-irrelevant nodes,
$\x_i^1 \in \Rpluseq^{n_i-r_i}$ the task-relevant nodes and $\B_i \in
\R^{n_i \times p_i}$ is such that the task-relevant nodes are not
impacted by the control term $\u_i \in \R^{p_i}$.  Throughout we shall
assume that $p_i \geq r_i$, and that the matrices $\B_i^0$ have all
full rank.

The thalamus layer $\mathcal{N}_T$ is modeled in the same way as the
cortical layers with the difference appearing in the interconnection
term $\w_T(t)$. A final word about the brain mechanisms for the inhibition of
regions~\cite{JSI-MS:11}. Feedforward inhibition between two layers,
$\mathcal{N}_i$ and $\mathcal{N}_j$, refers to when $\mathcal{N}_i$
sends an inhibitory signal to $\mathcal{N}_j$ to achieve its goal
activity pattern for $\mathcal{N}_j$ regardless of the current state
of $\mathcal{N}_j$. In contrast, feedback inhibition refers to when
the inhibition applied in a layer $\mathcal{N}_i$ is dependent upon
the current activity level of the nodes desired to be
inhibited~\cite{JSI-MS:11}.
We employ a combination of feedback and feedforward inhibition.

\subsection{Interconnection Topology Among Network
  Layers}\label{sec:interconnection}

We detail here the dynamical interconnection between the network
layers for each of the topologies depicted in
Fig.~\ref{fig:thalamocortical-networks}.

\subsubsection*{Hierarchical Thalamocortical Networks}
Consider the hierarchical thalamocortical network depicted in
Fig.~\ref{fig:thalamocortical-networks}(a).  The hierarchical
structure is encoded by the ordered timescales of the layers:
$\tau_1 \gg \tau_2 \gg \dots \gg \tau_N$, prescribing progressively
faster dynamics as one moves down the hierarchy. For generality, the
timescale $\tau_T$ of the thalamus might fit anywhere within the
hierarchy. For each cortical layer $\mathcal{N}_i$, the
interconnection term $\w_i(t)$
in~\eqref{eq:lin_threshold_input_expansion} takes the form
\begin{align}\label{eq:inputs_cortical}
  \w_i(t) &= \W_{i,i-1}\x_{i-1}(t) + \W_{i,i+1}\x_{i+1}(t) +
  \W_{i,T}\x_T(t) .
\end{align} 
Here, the terms $\W_{i,i-1}, \W_{i,i+1}$, and $\W_{i,T}$ represent the
weights of the synaptic connections between layers $\mathcal{N}_i$ and
$\mathcal{N}_{i-1}$, $\mathcal{N}_{i+1}$ and $\mathcal{N}_T$,
resp. It is important to note that since the thalamus impacts
the cortical regions using feedforward inhibition, the interconnection
matrix between the thalamus and the cortical layer
satisfies $\W_{i,T} \leq 0$.  Substituting~\eqref{eq:inputs_cortical}
into the linear-threshold dynamics~\eqref{eq:lin_threshold_dynamics},
we get the dynamics for a cortical layer $\mathcal{N}_i$
\begin{align}\label{eq:cortical_layer_dynamics}
  \tau_i\dot{\x}_i &= -\x_i + [\W_{i,i}\x_i + \W_{i,i-1}\x_{i-1}
  \\
  &\quad + \W_{i,i+1}\x_{i+1} + \W_{i,T}\x_T + \B_i\u_i(t) +
  \c_i]_\zero^{\m_i}, \notag
\end{align}
for $i \in \until{N}$. For consistency $\W_{1,0} = \zero = \W_{N,N+1}
$, and we assume that $r_1 = 0$, meaning no nodes are being inhibited
in the top layer of the network.

For the thalamus layer, to reflect the different connectivity it has
in the network, the interconnection term $\w_T(t)$ is
\begin{align}\label{eq:inputs_thalamus}
  \w_T(t) = \sum_{i=1}^N \W_{T,i}\x_i(t)  .
\end{align}
Here $\W_{T,i}$ represents the weight of the synaptic connections
between layers $\mathcal{N}_T$ and $\mathcal{N}_i$ for
$i \in \{1,\dots,N\}$. Then, substituting~\eqref{eq:inputs_thalamus}
into~\eqref{eq:lin_threshold_dynamics}, the dynamics for the thalamus
layer is
\begin{align}\label{eq:thalamus_layer_dynamics}
  \tau_T\dot{\x}_T &= -\x_T + [\W_{T}\x_T + \sum_{i=1}^N \W_{T,i}\x_i +
  \B_T\u_T(t) + \c_T]_\zero^{\m_T}.
\end{align}
We denote the timescale ratio between layers by $\epsilon =
(\epsilon_1,\dots,\epsilon_N,\epsilon_T)$, where
\begin{align*}
  \epsilon_i =
  \frac{\tau_{i}}{\tau_{i-1}},
  \quad \epsilon_T = \frac{\tau_T}{\min_{\tau_j > \tau_T}\tau_j}.
\end{align*}
For a subnetwork $\mathcal{N}_i$ such that the thalamus timescale fits
in the hierarchy directly above it, i.e. $\tau_T > \tau_i$ but there
does not exist $j$ such that $\tau_T > \tau_j > \tau_i$, we let
the timescale ratio between $\mathcal{N}_T$ and $\mathcal{N}_i$ be
given by $\bar{\epsilon}_i = {\tau_i}/{\tau_T}$.

\subsubsection*{Star-Connected Thalamocortical Networks}

Consider the star-connected thalamocortical network depicted in
Fig.~\ref{fig:thalamocortical-networks}(b). In contrast to the
hierarchical network, this topology does not form a hierarchical
timescale, and as such there is no direct relationship satisfied by
the timescales. The lack of an explicit relation encodes the thalamus'
role as a sensory relay to multiple brain regions, each part of
potentially unrelated temporal hierarchies.  Without loss of
generality, we assume subnetwork $\mathcal{N}_1$ represents a
subcortical structure and provides the input to the network for the
thalamus to relay to the other brain regions.

As such, there are no nodes in $\mathcal{N}_1$ that are desired to be
inhibited, meaning $r_1 = 0$. We model the subcortical input
subnetwork with a linear-threshold dynamics, but without an
independent control term, instead modeling input changes by allowing
$\c_1$ to be time-varying. That is
\begin{align}\label{eq:input_star_dynamics}
  \tau_1\dot{\x}_1 = -\x_1 + [\W_{1,1}\x_1 + \W_{1,T}\x_T +
  \c_1(t)]_\zero^{\m_1}. 
\end{align}
For the cortical regions $\mathcal{N}_i$, $i \in \{2,\dots,N\}$, the
interconnection term $\w_i(t)$ is given by $\w_{i}(t) = \W_{i,T}\x_T(t)$, and we
recall that as the thalamus utilizes feedforward inhibition, $\W_{i,T} \leq
0$. Meanwhile, the interconnection of the thalamus with the cortical regions is defined by $\w_T(t) =
\sum_{i=1}^N \W_{T,i}\x_i(t)$. Then, from the linear-threshold
model~\eqref{eq:lin_threshold_dynamics}, for $i \in \{2,\dots,N\}$ the
dynamics takes the form
\begin{align}\label{eq:star_dynamics}
  \tau_i\dot{\x}_i &= -\x_i + [\W_{i,i}\x_i + \W_{i,T}\x_T + \B_i\u_i(t) +
  \c_i]_\zero^{\m_i},  \\
  \tau_T\dot{\x}_T &= -\x_T + [\W_{T}\x_T +
  \sum_{i=1}^N \W_{T,i}\x_i + \B_T\u_T(t) + \c_T]_\zero^{\m_T}. \notag
\end{align}

\subsection{Problem Statement}\label{sec:problem-statement}
For a purely cortical hierarchical brain network, selective inhibition
and recruitment can be achieved using a combination of feedback and
feedforward control, dependent on the subnetwork dynamics satisfying a
set of stability properties,
cf.~\cite{EN-JC:21-tacI,EN-JC:21-tacII}. However, with the presence of
the thalamus and given its impact on the dynamics of the individual
subnetworks, such results are inapplicable. In addition, for
star-connected thalamocortical topology the results do not apply, as
the assumption of a hierarchical relationship between subnetwork
timescales, which plays a critical role, no longer holds. As such, in
this work we study the problem of selective inhibition and recruitment
for thalamocortical networks, both hierarchical and star-connected,
which we formalize next.

We seek to determine the existence of control mechanisms, resulting
from the combination of feedback and feedforward inhibition, such that
for a thalamocortical network, with either interconnection topology as
described in Section~\ref{sec:interconnection}, each of the individual
subnetworks $\mathcal{N}_i$ and $\mathcal{N}_T$ achieves selective
inhibition and recruitment. In particular, we look for a control law
of the form
\begin{align*}
  \u_i(t) = \bar{\u}_i(t) + \K_i\x_i(t),
\end{align*}
for each $i \in \{2,\dots,N,\}\cup\{T\}$, that results in the state of
each subnetwork converging to a desired equilibrium
trajectory. Solving this problem requires determining the stability
properties that the dynamics of each of the individual subnetworks
must satisfy to guarantee the existence of such a control law.

In addition, we seek to evaluate how the addition of the thalamus
results in improved performance relative to a purely cortical
network. In particular we investigate how the addition of the thalamus
impacts the control magnitude needed to achieve selective inhibition
and the convergence speed of the network. Further, we study
differences in the hierarchical and star-connected topologies,
particularly in how the latter can act as a failsafe for the former.

\section{Hierarchical Thalamocortical
  Networks}\label{sec:multilayer_networks}

In this section, we consider a hierarchical multilayer thalamocortical
network, cf. Fig.~\ref{fig:thalamocortical-networks}(a), where the
cortical layers are governed by~\eqref{eq:cortical_layer_dynamics} and
the thalamus layer is governed by~\eqref{eq:thalamus_layer_dynamics}.

\subsection{Equilibrium Maps for Individual Layers}\label{sec:eq_maps}

We note that, for a general linear-threshold
dynamics~\eqref{eq:lin_threshold_dynamics} with
$\W \in \R^{n \times n}$, its equilibrium map
$\map{h}{\R^n}{\Rpluseq^n}$
\begin{align}\label{eq:lin_thres_eq_map}
  h(\c) = h_{\W,\m}(\c) = \{\x \in \Rpluseq^n ~|~ \x = [\W\x +
  \c]_\zero^\m \} ,
\end{align}
maps a constant input $\c \in \R^n$ to the set of equilibria
of~\eqref{eq:lin_threshold_dynamics}.  For thalamocortical networks,
the timescale of the thalamus impacts the specific form of the
equilibrium maps.  For simplicity, we assume the thalamus lies inside
the hierarchy, i.e., there exist $a,b \in \until{N}$ such that
$b = a+1$ and $\tau_a \gg \tau_T \gg \tau_b$. This choice results in a
layer $\mathcal{N}_a$ with an equilibrium map different than any that
appear in the cases when the thalamus timescale is on the boundary of
the hierarchy.  Our results can also be stated for the latter case
with appropriate adjustments to the derived control law for selective
inhibition and recruitment.

The hierarchical nature of the topology plays an important role in
defining  these equilibrium maps.
At the theoretical limit of timescale separation between two layers,
the state of a given layer becomes constant at the timescale of layers
lower in the hierarchy as well as being a static function of the above
layer's states. As such, when defining the equilibrium maps for a
given layer, the constant input into an equilibrium map can represent
the state of layers higher in the hierarchy, given that they are
relative constants at that level of the dynamics.  The equilibrium
maps for the layers fall into three categories: below or above the
thalamus, and the thalamus itself. In all cases, the maps are defined
recursively from the bottom to the top of the network. At each layer,
the equilibrium map takes a constant input, representing the inputs
from higher levels in the hierarchy along with any external inputs to
the system, and outputs the set of equilibrium values.
We next give explicit expressions for the equilibrium maps of the
task-relevant component of a layer in each of the categories, which we
denote $\map{h_i^1}{\R^{n_i-r_i}}{\Rpluseq^{n_i-r_i}}$, by combining
the hierarchical model described in Section~\ref{sec:interconnection}
with~\eqref{eq:lin_thres_eq_map}.

\subsubsection*{Equilibrium Maps for Layers below Thalamus}
We begin by considering layers below the thalamus, i.e., with
$i \geq b$. Given values of $\x_T \in \R^{n_T}$ and
$\c_{i+1}^1 \in \R^{n_{i+1} - r_{i+1}}$,
\begin{align}\label{eq:below_thal_eq_map}
  h_{i}^1(\c) &= \{\x_i^1 ~|~ \x_i^1 = [\W_{i,i}^{11}\x_i^1 +
  \W_{i,i+1}^{11}h_{i+1}^1(\W_{i+1,i}^{11}\x_i^1 \notag
  \\
  &\qquad + \W_{i+1,T}^{11}\x_T^1 + \c_{i+1}^1) + \c]_\zero^{\m_i^1}
  \}.
\end{align}
We note that since $\W_{N,N+1} = \zero$ by convention, this recursion
is well-defined for $\mathcal{N}_N$, the bottom layer in the
network. In particular, for layer $N$,~\eqref{eq:below_thal_eq_map}
reduces to~\eqref{eq:lin_thres_eq_map}, the standard equilibrium map
for linear-threshold models.

\subsubsection*{Equilibrium Map for Thalamus}
Since the thalamus is connected to all the cortical layers, the
recursive definition of its equilibrium map is dependent on the
equilibrium maps of all of the layers below it. Due to this, using the
notation of~\eqref{eq:below_thal_eq_map} for writing the recursion
becomes intractable for large networks. Instead, we introduce the
following simplifying notation. For a layer $\mathcal{N}_i$, $i \in
\until{N}$, in the thalamocortical hierarchy, let $\x_{i_e}^1 \in
\Rpluseq^{n_i-r_i}$ be such that
\begin{align*}
  \x_{i_e}^1 \in h_i^1(\W_{i,T}^{11}\x_T^1 + \W_{i,i-1}^{11}\x_{i-1}^1
  + \c_i^1).
\end{align*}
Note that, despite $\x_{i_e}$ being dependent on a set of inputs, we
do not specify these. As such, when this notation is used, it is
implicitly assumed that the input values are given or determined in
the recursion. Depending on the point in the recursion being
considered, the inputs $\x_T$ and/or $\x_{i-1}$ will be replaced by
their equilibrium values. Now, given values $\c_b^1,\dots,\c_N^1$,
where $\c_i^1 \in \R^{n_i-r_i}$, the thalamus equilibrium map
$\map{h_T^1}{\R^{n_T-r_T}}{\Rpluseq^{n_T-r_T}}$~is
\begin{align}\label{eq:thal_eq_map}
  h_{T}^1(\c) = \{\x_T^1 ~|~ \x_T^1 = [\W_{T}^{11}\x_T^1 +
  \sum_{j=b}^N \W_{T,j}^{11}\x_{j_e}^1 + \c]_\zero^{\m_T^1} \} .
\end{align}

\subsubsection*{Equilibrium Maps for Layers above Thalamus}
We note that, in the recursive definition for the equilibrium maps
above the thalamus, the inputs to the thalamus equilibrium
map~\eqref{eq:thal_eq_map} will include equilibrium values of layers
above the thalamus in the hierarchy, in addition to the equilibrium
values from lower in the hierarchy. To distinguish which maps
are inputting equilibrium maps in a condensed manner, we introduce the
following notation. For $i \in \until{a}$, we let $\x_{T_{e(i)}}^1 \in
\R^{n_T-r_T}$ denote a value inside the thalamus equilibrium set
satisfying
\begin{align*}
  \x_{T_{e(i)}}^1 \in h_T^1(\sum_{j=1}^i \W_{T,j}^{11}\x_j^1 +
  \sum_{j=i+1}^a \W_{T,j}^{11}\x_{j_e}^1 + \c_T^1).
\end{align*}
Using this notation, we define the remaining equilibrium maps. 
We first provide the equilibrium maps for layers
$\{\mathcal{N}_i\}_{i=1}^{a-1}$ and finish with the map for
$\mathcal{N}_a$. For $i \in \until{a-1}$, given $\c_{i+1}^1 \in
\R^{n_i-r_i}$, the equilibrium map
$\map{h_i^1}{\R^{n_i-r_i}}{\Rpluseq^{n_i-r_i}}$~is
\begin{align}\label{eq:above_thal_eq_map}
  h_i^1(\c) &= \{ \x_i^1~|~ \x_i^1 = [\W_{i,i}^{1,\all}\x_i +
  \W_{i,i+1}h_{i+1}^1(\W_{i+1,i}^{1,\all}\x_i + \c_{i+1}^1) \notag
  \\
  &\qquad + \W_{i,T}^{11}\x_{T_{e(i)}} + \c]_\zero^{\m_i^1} \}.
\end{align}
The expression for the equilibrium map of the layer $\mathcal{N}_a$,
which is directly above the thalamus, differs from the other layers
due to the fact that it depends directly on a layer below the thalamus
in addition to being dependent on the thalamus. Given $\c_b^1 \in
\R^{n_b-r_b}$, the equilibrium map
$\map{h_a^1}{\R^{n_a-r_a}}{\Rpluseq^{n_a-r_a}}$~is
\begin{align}\label{eq:layer_a_eq_map}
  h_a^1(\c) &= \{\x_a^1 ~|~ \x_a^1 = [\W_{a,a}^{11}\x_a^1 +
  \W_{a,b}^{11}h_b^1(\W_{b,a}^{11}\x_a^1
  \\
  &\qquad+ \W_{b,T}^{11}\x_{T_{e(a)}}^{1}) + \c_b^1) +
  \W_{a,T}^{11}\x_{T_{e(a)}}^{1} + \c]_\zero^{\m_a^1} \}. \notag
\end{align}

We conclude this section by establishing that all the equilibrium maps
in the thalamocortical
network~\eqref{eq:below_thal_eq_map}-\eqref{eq:layer_a_eq_map} are
piecewise-affine and use this fact to justify they are globally
Lipschitz too.  To begin, we note that since general linear-threshold
dynamics are switched affine, their equilibrium
map~\eqref{eq:lin_thres_eq_map} can be written in a piecewise-affine
form. In particular, this equilibrium map can be written as follows
\begin{align}\label{eq:equilibrium_map_piecewise_affine}
  h(\c) &= \{\F_\sigma\c + \f_\sigma ~|~ \G_\sigma \c + \g_\sigma \geq
  \zero,~\sigma \in \{0,\ell,s\}^n\},
\end{align}
for some matrices and vectors of the form
\begin{align*}
  \F_\sigma &= (\eye - \Sigma^{\ell}\W)^{-1}, \qquad
  \f_\sigma = (\eye -\Sigma^{\ell}\W)^{-1}\Sigma^{\mathbf{s}}\m,
  \notag
  \\
  \G_\sigma &= \begin{bmatrix}\Sigma^{\ell} + \Sigma^{s}-\eye &
    \Sigma^{\ell} & -\Sigma^{\ell} & \Sigma^s \end{bmatrix}^\top
  \F_\sigma, \notag
  \\
  \g_\sigma &= \begingroup \setlength\arraycolsep{2pt}
  \begin{bmatrix}
    \f_\sigma^\top(\Sigma^{\ell}+\Sigma^s-\eye) &
    \f_\sigma^\top\Sigma^{\ell} & (\m-\f_\sigma)^\top\Sigma^{\ell} &
    (\f_\sigma-\m)^\top\Sigma^s \end{bmatrix}^\top,
  \endgroup \notag
\end{align*}
where $\Sigma^{\ell}$ is a diagonal matrix with $\Sigma^{\ell}_{ii} =
1$ if $\sigma_i = \ell$ and zero otherwise, and $\Sigma^{s}$ is
defined analogously. Now, since $\W_{N,N+1} = \zero$, the equilibrium
map for the bottom layer in the hierarchy, $h_N^1$, has the same form
as~\eqref{eq:lin_thres_eq_map}, and hence can be written in the
form~\eqref{eq:equilibrium_map_piecewise_affine}.
The following result, which is a generalization of~\cite[Lemma
IV.1]{EN-JC:21-tacII}, illustrates that in fact all the equilibrium
maps in the hierarchical thalamocortical network are piecewise-affine.

\begin{lemma}\longthmtitle{Piecewise-affinity of equilibrium maps in
    hierarchical thalamocortical linear-threshold
    models}\label{lemma:piecewise_affine_eq_maps}
  Let $\map{h_i}{\R^n}{\R^n}$, $i \in \{1,\dots,k\}$, be
  piecewise-affine functions,
  \begin{align*}
    h_i(\c) = \F^{i}_{\lambda_i}\c + \f^i_{\lambda_i}, \qquad \forall
    \c \in \Psi_{\lambda_i} \triangleq \{\c~|~\G_{\lambda_i}^i\c +
    \g_{\lambda_i}^i \geq \zero \} ,
  \end{align*}
  for $\lambda_i \in \Lambda_i$, where $\Lambda_i$ is a finite index
  set such that $\bigcup_{\lambda_i \in \Lambda_i} \Psi_{\lambda_i} =
  \R^n$. Define $\Lambda = \Lambda_1 \times \Lambda_2 \times \dots
  \times \Lambda_k$, $\lambda = (\lambda_1,\dots,\lambda_k)$ and
  $\Psi_\lambda = (\Psi_{\lambda_1},\dots,\Psi_{\lambda_k})$. Given
  matrices $\W_1,\W_2^i,\W_3^i$ and vectors $\bar{\c}_i,\c'$ for all
  $i \in \{1,\dots,k\}$ assume
  \begin{align}\label{eq:equilibrium_expression_piecewise_affine_lem}
    \x = \left[\W_1\x + \sum_{i=1}^k \W_2^i h_i(\W_3^i\x + \bar{\c}_i)
      + \c'\right]_\zero^\m,
  \end{align}
  is known to have a unique solution $\x' \in \R^{n'}$ for each $\c'
  \in \R^{n'}$, and let $h'(\c')$ be this solution. Then, there exists
  a finite index set $\Lambda'$ and
  $\{(\F_{\lambda}',\f_{\lambda}',\G_{\lambda}',\g_{\lambda}')\}_{\lambda'
    \in \Lambda'}$ such that
  \begin{align*}%%\label{eq:piecewise_affine_form}
    h'(\c') = \F_{\lambda'}'\c' + \f_{\lambda'}', \quad \forall \c'
    \in \Psi_{\lambda'}' \triangleq \{\c' ~|~ \G_{\lambda'}'\c' +
    \g_{\lambda'}' \geq \zero \}
  \end{align*}
  for $\lambda' \in \Lambda'$ and $\bigcup_{\lambda' \in \Lambda'}
  \Psi_{\lambda'} = \R^{n'}$.
\end{lemma}

The equilibrium maps for the hierarchical thalamocortical network
satisfy Lemma~\ref{lemma:piecewise_affine_eq_maps} with $r=1$, $r=N-b$
and $r=2$ for the layers below the thalamus, the thalamus, and the
layers above the thalamus, resp., and therefore the maps are
piecewise-affine. Finally, the fact that these maps are globally
Lipschitz follows from~\cite[Lemma IV.2]{EN-JC:21-tacII}. This
property is necessary to be able to apply later the generalization of
Tikhonov's singular perturbation stability theorem to non-smooth ODEs
given in~\cite[Proposition 1]{VV:97}\footnote{To apply the result to a
  non-smooth ODE such as~\eqref{eq:lin_threshold_dynamics} we need to
  justify the following: $1)$ Lipschitzness of the dynamics uniformly
  in $t$, $2)$ Existence, uniqueness and Lipschitzness of the
  equilibrium map of the fast dynamics, $3)$ Lipschitzness and
  boundedness of the reduced-order model, $4)$ Asymptotic stability of
  the fast dynamics uniformly in $t$ and the slow variable, and $5)$
  Global attractivity of the fast dynamics for any fixed slow
  variable.} to take advantage of the timescale separation between
the layers.

\subsection{Stability Assumptions and
  Conditions}\label{sec:stability_assumptions}
In the hierarchical thalamocortical model described in
Section~\ref{sec:problem-setup}, only the task-irrelevant components
of the dynamics are directly controlled over,
cf.~\eqref{eq:relevant2}. This means that assumptions on the stability
of the task-relevant components of the dynamics are needed to
guarantee their stability and recruitment to an equilibrium
trajectory. In particular, we are interested in the reduced-order
task-relevant dynamics, that is, the system dynamics in which the
inputs from layers lower in the hierarchy have been replaced by their
equilibrium values. Here, we provide details on these assumptions and
identify sufficient conditions for them to hold. We split our
discussion on the assumptions based upon the location of a layer in
the hierarchy as different locations maintain both different roles and
equilibrium maps.
\subsubsection*{Top Layer of the Hierarchy}
The top layer $\mathcal{N}_1$ in the hierarchy does not have any nodes
that are to be inhibited. In addition, its role is in driving the
selective recruitment in the lower levels, rather than being recruited
itself. As such, we only require that the trajectories of its dynamics
are bounded. Formally, for all sets of constants $\c_1 \in \R^{n_1}$,
$\c_i^1 \in \R^{n_i-r_i}$, $i \in \{2,\dots,N\}$, and
$\c_T \in \R^{n_T-r_T}$, we assume
\begin{align}
  \tau_1\dot{\bar{\x}}_1 &= -\x_1 + [\W_{1,1}\bar{\x}_1 +
  \W_{1,2}^{11}h_2^1(\W_{2,1}\bar{\x}_1 + \c_2^1) \notag
  \\
  &\qquad + \W_{1,T}^{11}\x_{T_{e(1)}} +
  \c_1]_\zero^{\m_1},\label{eq:reduced_dynamics_top}
\end{align}
has bounded solutions. We note that our earlier assumption that the
thalamus is in the middle of the hierarchy makes the top layer a
cortical one. If, instead, the thalamus was the top layer, one would
instead assume here that its dynamics has bounded solutions,
replacing~\eqref{eq:reduced_dynamics_top} accordingly.

\subsubsection*{Lower Layers in the Hierarchy}
In each of the layers below the top one, we seek to accomplish
selective inhibition and recruitment by having the dynamics converge
to a parameter-dependent equilibrium trajectory. For the
task-irrelevant components, we aim to use a control law to stabilize
them to zero. If successful, what remains is the task-relevant
dynamics, which we then need to be globally exponentially stable to a
parameter-dependent equilibrium trajectory. In such a case, an
additional challenge is then to identify conditions on the
interconnected network layers that ensure the coupled system displays
the desired behavior.  Due to the different reduced-order dynamics
throughout the network, we provide the details for this assumption in
four categories: layers above the thalamus, layer directly above the
thalamus, thalamus, and layers below the thalamus.

\subsubsection*{Layers above the Thalamus}
For each layer $\mathcal{N}_i$, with $i \in \until{a-1}$, above the
thalamus, the dynamics are directly dependent on two
layers below it, and hence the reduced-order dynamics are dependent on
two equilibrium maps. For achieving selective inhibition and
recruitment, we assume that, for all sets of constants
$\c_{i+1}^1 \in \R^{n_{i+1}-r_{i+1}}$, $\c_i^1 \in \R^{n_i-r_i}$, and
$\c_j^1 \in \R^{n_j-r_j}$, $j \in \{i+2,\dots,N\}\cup\{T\}$, the
reduced-order dynamics
\begin{align}
  \tau_i\dot{\x}_i^1 &= -\x_i^1 + [\W_{i,i}^{11}\x_i^1 +
  \W_{i,i+1}^{11}h_{i+1}^1(\W_{i+1,i}^{11}\x_i^1 \notag
  \\
  & \quad + \c_{i+1}^1) + \W_{i,T}^{11}\x_{T_{e(i)}} +
  \c_i^1]_\zero^{\m_i^1}, \label{eq:reduced_dynamics_above_thal}
\end{align}
are GES to an equilibrium trajectory defined by the chosen constants.

\subsubsection*{Layer Directly above the Thalamus}
For the layer $\mathcal{N}_a$ directly above the thalamus, while still
dependent on two layers below it in the hierarchy, one
of these layers is below the thalamus. In the reduced-order model,
this changes the equilibrium maps. Thus, for all constants $\c_{a}^1
\in \R^{n_a-r_a}$, $\c_b^1 \in \R^{n_b-r_b}$, and $\c_i^1 \in
\R^{n_i-r_i}$, $i \in \{b+1,\dots,N\}\cup\{T\}$, we assume its reduced-order
dynamics

\begin{align}
  \tau_a\dot{\x}_a^1 &= -\x_a^1 + [\W_{a,a}^{11}\x_a^1 +
  \W_{a,b}^{11}h_b^1(\W_{b,a}^{11}\x_a^1
  \notag \\
  & \quad + \W_{b,T}\x_{T_{e(a)}} + \c_b^1) +
  \W_{a,T}^{11}\x_{T_{e(a)}} +
  \c_a^1]_\zero^{\m_a^1}, \label{eq:reduced_dynamics_layer_a}
\end{align}
are GES to an equilibrium trajectory determined by the chosen constants.

\subsubsection*{Thalamus Layer}
Since the thalamus layer depends on all the other layers in the
network, its reduced-order dynamics are dependent on the $N-b$
equilibrium maps of all the layers (from $\mathcal{N}_b$ to
$\mathcal{N}_N$) below the thalamus.  For all constants $\c_b^1 \in
\R^{n_b-r_b},\dots,\c_N^1 \in \R^{n_N-r_N}$ and $\c_T^1 \in
\R^{n_T-r_T}$, we assume the reduced-order dynamics
\begin{align}
  \tau_T\dot{\x}_T^1 &= -\x_T^1 + \big[\W_T^{11}\x_T^1 + \sum_{j=1}^a
  \W_{T,j}^{1,\all}\x_j^1 \notag
  \\
  & \quad + \sum_{j=b}^N \W_{T,j}^{11}\x_{j_e}^1 + \c_T^1
  \big]_\zero^{\m_T^1}, \label{eq:reduced_dynamics_thal}
\end{align}
are GES to an equilibrium trajectory defined by the 
constants $\c_b^1,\dots,\c_N^1$ and $\c_T^1$.  

\subsubsection*{Layers below the Thalamus}
Finally, the layers below the thalamus are directly dependent on only
one layer below them in the hierarchy. For all constants $\c_{i+1}^1
\in \R^{n_{i+1}-r_{i+1}}$ and $\c_i^1 \in \R^{n_i-r_i}$, we assume its
reduced-order dynamics
\begin{align}
  \tau_i\dot{\x}_i^1 &= -\x_i^1 + [\W_{i,i}^{11}\x_i^1 +
  \W_{i,i+1}^{11}h_{i+1}^1(\W_{i+1,i}^{1,\all}\x_{i} \notag
  \\
  & \quad + \W_{i+1,T}^{1,\all}\x_T + \c_{i+1}^1) +
  \W_{i,T}^{1,\all}\x_T +
  \c_i^1]_\zero^{\m_i^1}, \label{eq:reduced_dynamics_below_thal}
\end{align}
and we assume they are GES to an equilibrium trajectory dependent upon the 
constants $\c_{i+1}^1$ and $\c_i^1$.

We next turn our attention to identifying conditions under which the
above assumptions on the individual subnetworks hold, specifically,
ensuring that the reduced-order dynamics for layers
$\mathcal{N}_2,\dots,\mathcal{N}_N,\mathcal{N}_T$ are GES.  The
following result makes use of the fact that the equilibrium maps are
piecewise-affine, cf. Lemma~\ref{lemma:piecewise_affine_eq_maps}.

\begin{lemma}\longthmtitle{Sufficient condition for existence and
    uniqueness of equilibria and GES in multilayer linear-threshold
    networks with parallel
    connections}\label{lemma:sufficient_GES_conditions}
  Let $\map{h_i}{\R^n}{\R^n}$, $i \in \{1,\dots,K\}$, be
  piecewise-affine functions,
  \begin{align*}
    h_i(\c) = \F^{i}_{\lambda_i}\c + \f^i_{\lambda_i}, \qquad \forall
    \c \in \Psi_{\lambda_i} \triangleq \{\c~|~\G_{\lambda_i}^i\c +
    \g_{\lambda_i}^i \} 
  \end{align*}
  for all $\lambda_i \in \Lambda_i$ where $\Lambda_i$ is a finite
  index set such that $\bigcup_{\lambda_i \in \Lambda_i}
  \Psi_{\lambda_i} = \R^n$. Define $\bar{\F}_i \triangleq
  \max_{\lambda_i \in \Lambda_i} |\F^i_{\lambda_i}|$ as the matrix
  made of the entry-wise maximum of the elements in
  $\{|\F_{\lambda_i}^i|\}_{\lambda_i \in \Lambda_i}$. For $i \in
  \{1,\dots,K\}$ let the matrices $\W_1^i,\W_2^i,$ be arbitrary and
  also consider arbitrary matrix $\W$. Then, if
  \begin{align*}
    \rho\left(|\W| + \sum_{i=1}^K |\W_1^i|\bar{\F}_i|\W_2^i|\right) < 1,
  \end{align*}
  the dynamics
  \begin{align*}
    \tau\dot{\x} &= -\x + \left[\W\x + \sum_{i=1}^K \W_1^i
      h_i(\W_2^i\x + \bar{\c}_i) + \c \right]_\zero^\m
  \end{align*}
  is GES to a unique for all constants $\bar{\c}_i,\c$.
\end{lemma}

Lemma~\ref{lemma:sufficient_GES_conditions} generalizes~\cite[Theorem
IV.4]{EN-JC:21-tacII} to thalamocortical networks and its proof
follows a similar line of arguments. The application of this result to
our setting results in, if
$ \rho(|\W_{i,i}^{11}| +
|\W_{i,i+1}^{11}|\bar{\F}_{i+1}|\W_{i+1,i}^{11}| +
|\W_{i,T}^{11}|\bar{\F}_T|\W_{T,i}^{11}|) < 1 $ for
$i \in \{2,\dots,a-1\}$,
$ \rho(|\W_{i,i}^{11}| +
|\W_{i,i+1}^{11}|\bar{\F}_{i+1}|\W_{i+1,i}^{11}|) < 1$, for
$i \in \{b,\dots,N\}$ and
$ \rho(|\W_{T}^{11}| + \sum_{i=b}^{N}
|\W_{T,i}^{11}|\bar{\F}_{i}|\W_{i,T}^{11}|) < 1$, then the
dynamics~\eqref{eq:reduced_dynamics_above_thal}-\eqref{eq:reduced_dynamics_below_thal}
are GES to an equilibria.

\begin{remark}\longthmtitle{Comparison with conditions for a strictly
    cortical hierarchical network} 
  Sufficient conditions for GES of the reduced-order dynamics of a
  strictly cortical network can be obtained from the above by allowing
  $\W_{i,T}^{11} = 0$ for all $i$, and then these conditions reduce to
  those found in~\cite{EN-JC:21-tacII}, as expected. These sufficient
  conditions to guarantee GES, for layers above the thalamus, are
  harder to satisfy in the thalamocortical network than in a cortical
  one. However, as the conditions are only sufficient, this does not
  mean that it is in fact more difficult to achieve GES of the
  reduced-order dynamics in the thalamocortical case, as the proof
  above does not explicitly invoke the inhibitory nature of the
  thalamus.  \oprocend
\end{remark}

\subsection{Selective Inhibition and Recruitment}
We are ready to illustrate how selective inhibition and recruitment
can be achieved in the hierarchical thalamocortical network
model. Here, we first formalize the concept mathematically and then
provide a feedforward-feedback control that achieves it.
Recall that selective inhibition corresponds to the task-irrelevant
components of the network converging to zero, and recruitment
corresponds to having the task-relevant components converge to an
equilibrium trajectory.  The timescale ratio between layers,
$\epsilon$, must approach zero to encode the separation of timescales
observed in the brain. As such, we require convergence of the
task-relevant components of the network to an equilibrium as this
ratio approaches zero.
Formally, selective inhibition and recruitment is achieved if the
following equations are satisfied for any $0 < t_1 < t_2 < \infty$:
first, for all layers $\mathcal{N}_i$, $i \in \{2,\dots,N\}\cup\{T\}$, it
holds that
\begin{subequations}\label{eq:selective-recruitment-inhibition}
  \begin{align}\label{eq:hier_selective_inhib}
    \text{[inhibition]:} & \quad \lim_{\epsilon \to 0} \sup_{t \in
      [t_1,t_2]} \norm{\x_i^0(t)} = 0;
    \intertext{second, for the top
      layer $\mathcal{N}_1$ in the hierarchy,}
    \label{eq:hier_selective_recr_top}
    \text{[driving layer]:} & \quad \lim_{\epsilon \to 0}
    \sup_{t \in [t_1,t_2]} \norm{\x_1^1(t) - \bar{\x}_1^1(t)} = 0;
    \intertext{for all layers $\{\mathcal{N}_i\}_{i=2}^a$ above the
      thalamus,}
    \label{eq:hier_selective_recr_above_thal}
    \text{[recruitment]:} & \quad \lim_{\epsilon \to 0} \sup_{t \in
      [t_1,t_2]} \Vert \x_i^1(t) \notag
    \\
    & \quad \quad - h_i^1(\W_{i,i-1}^{11}\x_{i-1}^1(t) + \c_i^1) \Vert =
    0;
    \intertext{for the thalamus layer,}
    \label{eq:hier_selective_recr_thal}
    \text{[recruitment]:} & \quad \lim_{\epsilon \to 0} \sup_{t \in
      [t_1,t_2]} \norm{\x_T^1(t) - x_{T_{e(a)}}} = 0;
    \intertext{and finally, for the layers $\{\mathcal{N}_i\}_{i=b}^N$
      below the thalamus,}
    \label{eq:hier_selective_recr_below_thal}
    \text{[recruitment]:} & \quad \lim_{\epsilon \to 0} \sup_{t \in
      [t_1,t_2]} \Vert\x_i^1(t) - h_i^1(\W_{i,i-1}^{1,\all}\x_{i-1}(t)
    \notag \\
    & \quad \qquad + \W_{i,T}^{1,\all}\x_T(t) + \c_i^1)\Vert = 0.
  \end{align}
\end{subequations}
Intuitively, we note from~\eqref{eq:selective-recruitment-inhibition}
that achieving selective inhibition and recruitment means that one can
make the error between the network trajectories and the equilibrium
trajectories arbitrarily small if the timescale ratio is small
enough. The following result shows that selective inhibition and
recruitment can be achieved in the hierarchical thalamocortical
network by means of a combination of feedforward and feedback control.

\begin{theorem}\longthmtitle{Selective inhibition and recruitment in  hierarchical
    thalamocortical networks with
  Feedforward-Feedback Control}\label{thrm:multilayer_LTR}
Consider an $N$-layer thalamocortical network as shown in
Fig.~\ref{fig:thalamocortical-networks}(a), with layer dynamics
given by~\eqref{eq:cortical_layer_dynamics}
and~\eqref{eq:thalamus_layer_dynamics}. Without loss of generality,
let $\tau_1 \gg \tau_2 \gg \dots \gg \tau_T \gg \dots \gg \tau_N$ and
$a,b \in \{1,\dots,N\}$ such that $\tau_a \gg \tau_T \gg \tau_b$, with
$b=a+1$. Assume the stability
conditions~\eqref{eq:reduced_dynamics_top}-\eqref{eq:reduced_dynamics_below_thal}
for the reduced-order subnetworks are satisfied.  Then, for $i \in
\{1,\dots,N\}\cup\{T\}$ and constants $\c_i \in \R^{n_i}$ and $\c_T \in
\R^{n_T}$, there exist control laws $\u_i(t) = \K_i\x_i(t) +
\bar{\u}_i(t)$, with $\K_i \in \R^{p_i \times n_i}$ and
$\map{\bar{\u}_i}{\Rpluseq}{\Rpluseq^{p_i}}$, such that the
closed-loop system achieves selective inhibition and
recruitment~\eqref{eq:selective-recruitment-inhibition}.
\end{theorem}
\begin{proof}
  We prove the result by constructing a control and iteratively
  applying the generalization of Tikhonov's theorem
  from~\cite{VV:97}. Throughout the proof, we make use of the
  following notation. For $i \leq a$, let $\x_{1:i}
  = \begin{bmatrix}\x_1^\top & \dots & \x_i^\top \end{bmatrix}^\top$
  and, for $i \geq b$, let $\x_{1:T:i} = \begin{bmatrix} \x_1^\top
    \dots \x_a^\top & \x_T^\top & \x_b^\top & \dots &
    \x_i^\top \end{bmatrix}^\top$. We first define the control for
  layer $\mathcal{N}_N$. Let $\K_N$ and $\bar{\u}_N(t)$ be such that
  \begin{subequations}\label{eq:multilayer_layer_N_control}
    \begin{align}
      \B_N^0\K_N &\leq
      -\W_{N,N}^{0,\all}, \label{eq:multilayer_layer_N_feedback}
      \\
      \B_N^0\bar{\u}_N(t) & \leq -\W_{N,N-1}^{0,\all} \x_{N-1}(t) -
      \W_{N,T}^{0,\all}\x_T(t) -
      \c_N^0. \label{eq:multilayer_layer_N_feedforward}
    \end{align}
  \end{subequations}
  The inequalities~\eqref{eq:multilayer_layer_N_control} can be
  satisfied due to our assumption that the matrices $\B_i^0$ have full
  rank and $p_i \geq r_i$ for all $i \in \{1,\dots,N\}\cup\{T\}$.
  Substituting~\eqref{eq:multilayer_layer_N_control} into the dynamics,
  \begin{align*}
    \tau_1\dot{\x}_1 &= -\x_1 + [\W_{1,1}\x_1 + \W_{1,2}\x_2 +
    \W_{1,T}\x_T + \c_1]_\zero^{\m_1},
    \\
    &\vdots
    \\
    \tau_T\dot{\x}_T &= -\x_T + \bigg[\W_T\x_T + \sum_{j=1}^N
    \W_{T,j}\x_j + \B_T\u_T  + \c_T \bigg]_\zero^{\m_T}, 
    \\
    &\vdots
    \\
    \epsilon_N\tau_{N-1}\dot{\x}_N^0 &= -\x_N^0, \\
    \epsilon_{N}\tau_{N-1}\dot{\x}_N^1 &= -\x_N^1 +
    [\W_{N,N}^{1,\all}\x_N + \W_{N,N-1}^{1,\all}\x_{N-1}
    \\
    &\qquad + \W_{N,T}^{1,\all}\x_T + \c_N^1]_{\zero}^{\m_1^1}.
  \end{align*}
  Taking $\epsilon_N \to 0$ then provides a separation of
  timescales between $\x_N$ and $\x_{1:T:N-1}$.  Note that, by
  assumption, the reduced-order
  dynamics~\eqref{eq:reduced_dynamics_below_thal} is GES.  Using then
  the fact that the cascaded interconnection of a GES system with an
  exponentially vanishing system is also GES, cf.~\cite[Lemma
  A.1]{EN-JC:21-tacI}, we deduce that, for any constants $\x_{N-1}$
  and $\x_T$, $\x_N$ is GES to
  $(\zero_{r_N},h_N^1(\W_{N,N-1}^{1,\all}\x_{N-1} +
  \W_{N,T}^{1,\all}\x_T + \c_N^1))$. Recalling that, by
  Lemma~\ref{lemma:piecewise_affine_eq_maps} and~\cite[Lemma
  IV.2]{EN-JC:21-tacII}, the equilibrium maps $h_i^1$ are globally
  Lipschitz for all $i \in \{1,\dots,N\}\cup\{T\}$, and noting that the
  entire network is Lipschitz due to the Lipschitzness of the
  linear-threshold function $[\cdot]_\zero^\m$, we can
  apply~\cite[Proposition 1]{VV:97}, giving for $0 < t_1 < t_2 <
  \infty$,
  \begin{align}\label{eq:limits_bottom_layer}
    & \lim_{\epsilon_1 \to 0} \sup_{[t_1,t_2]} \norm{\x_N^0(t)} = 0,
    \notag \\
    & \lim_{\epsilon_1 \to 0} \sup_{[t_1,t_2]} \big\Vert\x_N^1(t) -
      h_N^1(\W_{N,N-1}^{1,\all}\x_{N-1}(t) \notag \\ 
      & \qquad + \W_{N,T}^{1,\all}\x_T(t) +
      \c_N^1)\big\Vert =0,
    \notag \\
    & \lim_{\epsilon_1 \to 0} \sup_{[t_1,t_2]} \norm{\x_{1:T:N-1} -
      \x_{1:T:N-1}^{(1)}} = 0,
  \end{align}
  where $\x_{1:T:N-1}^{(1)}$ represents the first-step reduced-order
  model coming from replacing $\x_N$ by its equilibrium value
  $(\zero_{r_N},h_N^1(\W_{N,N-1}^{1,\all}\x_{N-1} +
  \W_{N,T}^{1,\all}\x_T + \c_N^1))$. We continue the proof by
  iterating the above process for each network layer, utilizing
  constructed control laws and applying~\cite[Proposition 1]{VV:97} to
  the corresponding reduced-order models built from substitution of
  the equilibrium values.

  We now construct control laws for layers $\mathcal{N}_i$ with $i \in
  \{b,\dots,N-1 \}$. We choose $\K_i$ and $\bar{\u}_i(t)$ such that
  \begin{subequations}\label{eq:multilayer_below_thalamus_control}
    \begin{align}
      \B_i^0\K_i &\leq -|\W_{i,i}^{0,\all}| -
      |\W_{i,i+1}^{01}|\bar{\F}_{i+1}|\W_{i+1,i}^{1,\all}|, 
      \label{eq:below_thalamus_multilayer_feedback}
      \\
      \B_i^0\bar{\u}_i(t) &\leq -\W_{i,i-1}^{0,\all}\x_{i-1}(t) -
      \W_{i,T}^{0,\all}\x_T(t) - \c_i^0 \cr &\qquad-
      |\W_{i,i+1}^{01}|\bar{\F}_{i+1}|\W_{i+1,T}^{1,\all}\x_T(t) +
      \c_{i+1}^1|, \label{eq:below_thalamus_multilayer_feedforward}
    \end{align}
  \end{subequations}
  in which $\bar{\F}_i \in \R^{(n_i-r_i)\times(n_i -r_i)}$ is the
  entry-wise maximal gain of the equilibrium map $h_i^1(\cdot)$ as
  defined in Lemma~\ref{lemma:sufficient_GES_conditions}.
  We use the control law~\eqref{eq:multilayer_below_thalamus_control} to construct the
  reduced-order model, consider the timescale separation by letting
  $\epsilon_i \to 0$ for $i \in \{b,\dots,N-1 \}$, and finally
  apply~\cite[Proposition 1]{VV:97} to obtain

  \begin{align*}
    & \lim_{\epsilon \to 0}\sup_{t \in [t_1,t_2]}
    \norm{\x_i^{(N-i)^0}(t)} = 0,
    \\
    & \lim_{\epsilon \to 0}\sup_{t \in [t_1,t_2]}
    \Vert\x_i^{(N-i)^1}(t) -
    h_i^1(\W_{i,i-1}^{1,\all}\x_{i-1}^{(N-i)}(t)
    \\
    & \qquad + \W_{i,T}^{1,\all}\x_T^{(N-i)}(t)+ \c_i^1)\Vert = 0,
    \\
    & \lim_{\epsilon \to 0}\sup_{t \in [t_1,t_2]}
    \norm{\x_{1:T:i-1}^{(N-i)}(t) - \x_{1:T:i-1}^{(N-i+1)}(t)} = 0,
  \end{align*}
  for all $i \in \{b,\dots,N-1\}$. For the thalamus layer, we note
  that provided the initial conditions lie in $[\zero,\m_i]$ for all
  $i \in \{1,\dots,N\}\cup\{T\}$, by the properties of the linear-threshold
  dynamics, we have that $\x_i(t) \leq \m_i$ for all $t \geq 0$.
  Utilizing these bounds, we define the control for the thalamus layer
  such that it satisfies
  \begin{subequations}\label{eq:multilayer_thalamus_control}
    \begin{align}
      \B_T^0\K_T &\leq -|\W_T^{0,\all}| - \sum_{j=b}^N
      |\W_{T,j}^{01}|\bar{\F}_j|\W_{j,T}^{0,\all}|, \label{eq:multilayer_thalamus_feedback} 
      \\
      \B_T^0\bar{\u}_T(t) &\leq \sum_{j=1}^1 \W_{T,j}^{0,\all}\x_j(t) -
      \c_T^1 \cr &\quad -
      |\W_{T,b}^{0,\all}|\bar{\F}_b|\W_{b,a}^{1,\all}\x_a(t) + \c_b^1|
      \cr &\quad- \sum_{j=b+1}^{N}
      |\W_{T,j}^{01}|\bar{F}_j|\W_{j,j-1}^{1,\all}\m_{j-1} +
      \c_j^1|. \label{eq:multilayer_thalamus_feedforward}
  \end{align}
  \end{subequations}

  Then, after constructing the reduced-order model using the control
  laws~\eqref{eq:multilayer_thalamus_feedback}
  and~\eqref{eq:multilayer_thalamus_feedforward} and letting
  $\epsilon_T \to 0$ to create the timescale separation, we again
  apply~\cite[Proposition 1]{VV:97} to get
  \begin{align*}
    & \lim_{\epsilon \to 0} \sup_{t \in [t_1,t_2]}
    \norm{\x_T^{(T)^0}(t)} = 0,
    \\
    & \lim_{\epsilon \to 0} \sup_{t \in [t_1,t_2]}
    \norm{\x_T^{(T)^1}(t) - h_T^1(\sum_{j=1}^a
      \W_{T,j}^{1,\all}\x_j^{(T)}(t) + \c_T^1)} = 0,
    \\
    & \lim_{\epsilon \to 0} \sup_{t \in [t_1,t_2]}
    \norm{\x_{1:a}^{(T)}(t) - \x_{1:a}^{(N-b+1)}(t)} = 0.
  \end{align*}
  What remains is to consider the layers above the thalamus. These
  layers maintain the same form of control law except for the layer
  immediately above the thalamus. For layer $\mathcal{N}_a$ we
  define terms $\K_a$ and $\bar{\u}_a(t)$ such that
  \begin{subequations}\label{eq:multilayer_layer_a_control}
    \begin{align} 
      \B_a^0\K_a &\leq -|\W_{a,a}^{0,\all}| -
      |\W_{a,b}^{01}|\bar{\F}_{b}|\W_{b,a}^{1,\all}|,
      \\%\label{eq:multilayer_layer_a_feedback}\\
      \B_a^0\bar{\u}_a(t) &\leq -\W_{a,a-1}^{0,\all}\x_{a-1}(t) -
      |\W_{a,b}^{01}|\bar{\F}_{b}|\W_{b,T}^{1,\all}\x_{T_{e(a)}}+\c_{b}^1|.  % \label{eq:multilayer_layer_a_feedforward}
    \end{align}
  \end{subequations}
  Now, for layers $\mathcal{N}_i$, $i \in \{2,\dots,a-1 \}$, we let
  the controls $\K_i$ and $\bar{\u}_i(t)$ be such that the following
  hold
  \begin{subequations}\label{eq:multilayer_above_thalamus_control}
    \begin{align}
      \B_i^0\K_i &\leq -|\W_{i,i}^{0,\all}| -
      |\W_{i,i+1}^{01}|\bar{\F}_{i+1}|\W_{i+1,i}^{1,\all}|,
      \\ % \label{eq:multilayer_above_thalamus_feedback}\\
      \B_i^0\bar{\u}_i(t) &\leq -\W_{i,i-1}^{0,\all}\x_{i-1}(t) -
      |\W_{i,i+1}^{01}|\bar{\F}_{i+1}|\c_{i+1}^1| , %\label{eq:multilayer_above_thalamus_feedforward}
    \end{align}
  \end{subequations}

  Once the reduced-order models are constructed, letting $\epsilon_i
  \to 0$ for $i \in \{2,\dots,a\}$, applying~\cite[Proposition
  1]{VV:97} gives
  \begin{align*}
    &\lim_{\epsilon \to 0} \sup_{t \in [t_1,t_2]}
    \norm{\x_i^{(N-i)^0}(t)} = 0,
    \\
    &\lim_{\epsilon \to 0} \sup_{t \in [t_2,t_2]}
    \norm{\x_i^{(N-i)^1}(t) -
      h_i^1(W_{i,i-1}^{1,\all}\x_{i-1}^{(N-i)}(t) + \c_i^1 )} = 0,
    \\
    &\lim_{\epsilon \to 0}\sup_{t \in [t_1,t_2]}
    \norm{\x_{1:i-1}^{(N-i)}(t) - \x_{1:i-1}^{(N-i+1)}(t)} = 0.
  \end{align*}
  for $i \in \{2,\dots,a\}$. The
  equations~\eqref{eq:selective-recruitment-inhibition} for selective
  inhibition and recruitment, are then obtained through repeated
  application of the triangle inequality, completing the proof.
\end{proof}

\begin{figure}[htb]
  \centering
  \includegraphics[scale=0.18]{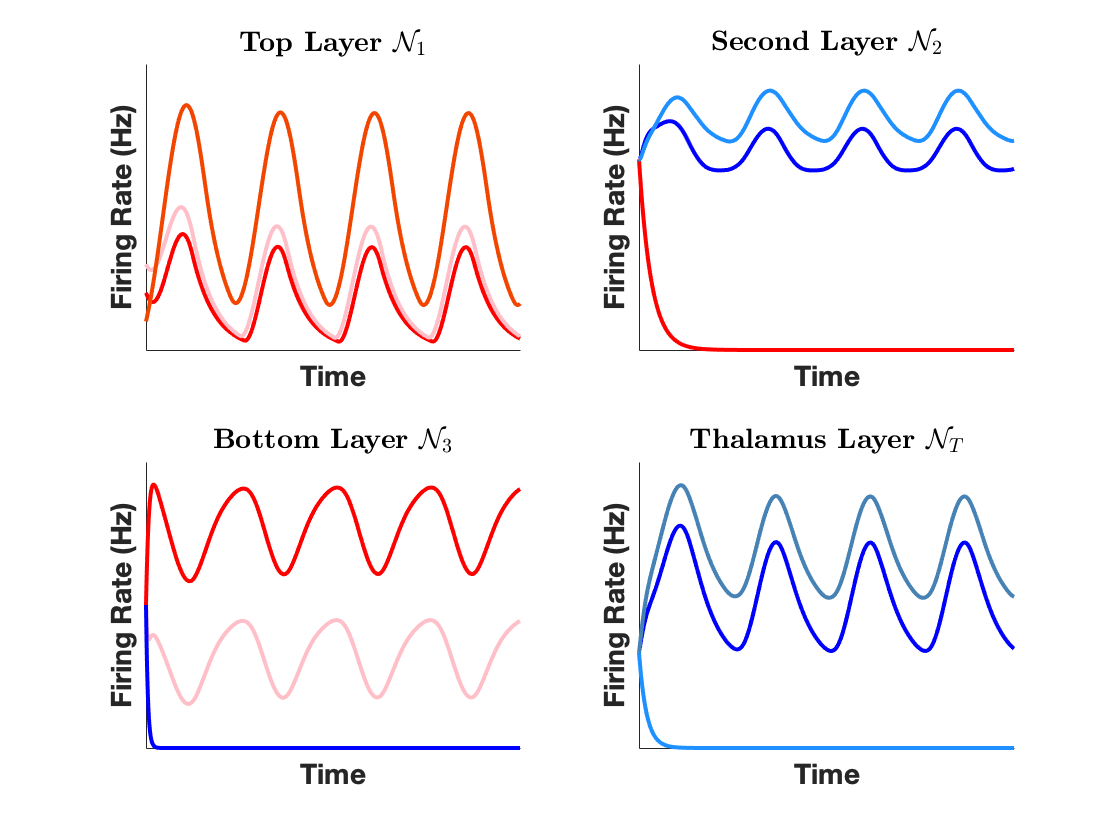}
  \caption{Trajectories of a three-layer hierarchical thalamocortical
    network under a periodic input.  Red and blue lines correspond
    with excitatory and inhibitory nodes, resp. Each layer
    consists of three nodes with the top layer in the network
    containing only excitatory nodes, the thalamus being strictly
    inhibitory and layers two and three having both excitatory and
    inhibitory nodes. Using a control law designed through
    Theorem~\ref{thrm:multilayer_LTR}, in each layer other than the
    top layer, nodes are able to be either selectively inhibited to
    zero or recruited into a periodic equilibrium trajectory.
    Timescale ratios are $\epsilon_2 = 0.54$, $\epsilon_3 = 0.15$ and
    $\epsilon_T = 0.61$.}\label{fig:sample_trajectories}
  \vspace*{-1ex}
\end{figure}

Fig.~\ref{fig:sample_trajectories} illustrates
Theorem~\ref{thrm:multilayer_LTR}. While the result is proven for the
case in which the thalamus has a timescale inside the hierarchy, it
still holds when the thalamus timescale is at the top or bottom of the
hierarchy. In these cases, the proof method remains the same, with the
appropriate modifications to the inequalities derived
in~\eqref{eq:multilayer_layer_N_control}-\eqref{eq:multilayer_above_thalamus_control}.

\begin{remark}\longthmtitle{Comparison with strictly cortical
    networks}
  Regarding selective inhibition and recruitment,
  Theorem~\ref{thrm:multilayer_LTR} is to multilayer thalamocortical
  networks what~\cite[Theorem IV.3]{EN-JC:21-tacII} is to multilayer
  cortical networks.  Despite the analytical similarities, the
  consideration of the thalamus provides a significant generalization
  from a biological perspective, as transthalamic connections exist in
  most brain networks~\cite{SMS:12}.  From a technical viewpoint, the
  addition of the thalamus, while only adding a layer, significantly
  complicates the analysis due to its connection with all of the
  cortical layers. These connections result in every layer having
  connections from timescales not simply immediately above or below it
  in the hierarchy, which impacts the determination of convergence to
  equilibrium values for layers above the thalamus. Finally,the control
  laws~\eqref{eq:multilayer_layer_N_control}-\eqref{eq:multilayer_above_thalamus_control}
  allow for smaller magnitude sufficient controls than for the
  strictly cortical networks determined in~\cite{EN-JC:21-tacII} due
  to the inhibitory properties of the thalamic connection
  matrices~$\W_{i,T}$, cf Section~\ref{sec:examples}.  \oprocend
\end{remark}

\section{Star-connected Thalamocortical
  Networks}\label{sec:star_networks}

In this section we consider star-connected thalamocortical networks,
cf. Fig.~\ref{fig:thalamocortical-networks}(b), where the cortical
regions are each connected only to the thalamus and the dynamics are
governed by~\eqref{eq:input_star_dynamics}-\eqref{eq:star_dynamics}.
In this topology, there is no direct relationship between the
timescales of each layer and as such, no hierarchical structure.
Without timescale separation, the specifics of selective inhibition
and recruitment, both in terms of equilibria and stability criteria,
differ from the hierarchical case.

\subsection{Equilibria and Stability Conditions}
With the lack of timescale separation, the decomposition described in
Section~\ref{sec:eq_maps} of the network equilibrium map as a
collection of equilibrium maps for each layer, with the state of
layers higher in the hierarchy represented by a constant input, no
longer holds.  As such, equilibrium values must be determined
concurrently for all the layers. Since the task-irrelevant components
get selectively inhibited to zero, the equilibrium for the
task-relevant components is given by the solution $\x_i^{1^*}$,
$i \in \until{N} \cup \{T\}$ to the system of equations:
\begin{align}\label{eq:equilibrium_sys_equations_star}
  \x_i^1 &= [\W_{i,i}^{1,\all}\x_i + \W_{i,T}^{1,\all}\x_T +
  \c_i^1]_\zero^{\m_i^1}, \qquad i = 1,\dots,N, \notag 
  \\
  \x_T^1 &= [\W_T\x_T + \sum_{i=1}^N \W_{T,i}\x_i +
  \c_T^1]_0^{\m_T^1}.
\end{align}

The task-relevant equilibrium to which the system converges is
dependent upon $\c_i \in \R^{n_i - r_i}$, $i \in \{1,\dots,N\}$ and
$\c_T \in \R^{n_T-r_T}$, and so we represent it by
$\x_i^{1^*}(\c_1,\dots,\c_N,\c_T)$.  In keeping with the role of layer
$\mathcal{N}_1$ as driving the selective recruitment in the other
layers, rather than being recruited itself, we consider in what
follows an input signal $\map{\c_1}{\Rpluseq}{\R^{n_i-r_i}}$, rather
than a constant, that gives rise to an equilibrium trajectory
$\x_i^{1^*}(\c_1(t),\dots,\c_N,\c_T)$ for the dynamics.

As per the description of the star-connected thalamocortical network,
cf Section~\ref{sec:problem-setup}, only the task-irrelevant component
of the dynamics is directly controlled, and as such the ability to
achieve selective inhibition and recruitment is dependent on the
stability properties of the task-relevant components. To ensure this,
we employ below the fact~\cite[Theorem IV.8]{EN-JC:21-tacI} that, for
a generic linear-threshold network mode
$\tau\dot{\x} = -\x + [\W\x + \c]_\zero^\m$, the condition
$\rho(|\W|) < 1$ is sufficient to ensure that, for all $\c \in \R^n$,
the dynamics is GES to an equilibrium.

\subsection{Selective Inhibition and Recruitment}
We are ready to formalize selective inhibition and recruitment for
star-connected networks and provide conditions for its achievement.
We recall that the subnetwork $\mathcal{N}_1$ corresponds with a
subcortical region applying a sensory input signal to the thalamus to
be relayed to the cortical regions. As such, we do not inhibit any
components in this subnetwork and instead assume that it is stable to
a trajectory $\bar{\x}_1(\c_1(t))$ dictated by its own input
signal. For the remaining layers, we wish to inhibit the
task-irrelevant components to zero and recruit the task-relevant
components to the equilibria trajectory
$\x_i^{1^*}(\c_1(t),\dots,\c_N,\c_T)$,
$i \in \{2,\dots,N\} \cup \{T\}$. This can be formalized to selective
inhibition and recruitment is achieved if for the input layer
$\mathcal{N}_1$,
\begin{subequations}\label{eq:selective_inh_rec_star}
  \begin{align} 
    \text{[driving layer]:} & \quad \lim_{t \to \infty} \norm{\x_1(t)
      - \bar{\x}_1(\c_1(t))} = 0; % \x_1^*(\c_1,\dots,\c_N,\c_T);
    \intertext{and for all layers $\{\mathcal{N}_i\}_{i=2}^{N}$ and
      $\mathcal{N}_T$,} \text{[inhibition]:} & \quad \lim_{t \to
      \infty} \norm{\x_i^0(t)} = 0; \label{eq:selective_inhib_star}
    \\
    \text{[recruitment]:} & \quad \lim_{t \to \infty} \norm{\x_i^1(t)
      - \x_i^{1^*}(\c_1(t),\dots,\c_N,\c_T)} = 0.
    \label{eq:selective_recr_star} %
  \end{align}
\end{subequations}
We also employ a weaker notion, referred to as $\epsilon$-selective
inhibition and recruitment, which is met if there exists $t^*$ such
that the functions in~\eqref{eq:selective_inh_rec_star} are all upper
bounded by $\epsilon>0$ for $t > t^*$.
For convenience, we also introduce the notation:
\begin{align}\label{eq:Wbar_matrix}
  \bar{\W}^{11} &= \begin{bmatrix}
    \W_{1,1} & \zero & \dots & \zero & \W_{1,T}^{11} \\
    \zero & \W_{2,2}^{11} & \dots & \zero & \W_{2,T}^{11} \\
    \vdots & \dots & \ddots & \vdots & \vdots \\
    \zero & \zero & \dots & \W_{N,N}^{11} & \W_{N,T}^{11} \\
    \W_{T,1}^{11} & \W_{T,2}^{11} & \dots & \W_{T,N}^{11} &
    \W_{T}^{11}
  \end{bmatrix}.
\end{align}
Note the Schur decomposition
$\bar{\W}^{11} = \mathbf{Q}^\top(\mathbf{D}_{\bar{\W}^{11}} +
\mathbf{N}_{\bar{\W}^{11}})\mathbf{Q}$, where $\mathbf{Q}$ is unitary,
$\mathbf{D}_{\bar{\W}^{11}}$ is diagonal, and
$\mathbf{N}_{\bar{\W}^{11}}$ is upper triangular with a zero
diagonal~\cite{KEC:86}.  The next result establishes conditions to
achieve selective inhibition and recruitment in star-connected systems
without a hierarchy of timescales.

\begin{theorem}\longthmtitle{Selective inhibition and recruitment of
   star-connected
    networks}\label{thrm:star_uniform_timescale}
  Consider an $N$-layer star-connected thalamocortical network as
  shown in Fig.~\ref{fig:thalamocortical-networks}(b), with layer
  dynamics given by~\eqref{eq:input_star_dynamics}
  and~\eqref{eq:star_dynamics}. Suppose the following hold for all
  values of $\c_i \in \R^{n_i}$, $i \in \until{N}$, and
  $\c_T \in \R^{n_T}$:
  \begin{enumerate}[(i)]
  \item The input layer $\mathcal{N}_1$ has no nodes to be inhibited,
    $\rho(|\W_{1,1}|) < 1$, and 
    % for all $\c_1 \in \R^{n_1}$, the dynamics
    % \begin{align*}
    %   \tau_1\dot{\x}_1 = -\x_1 + [\W_{1,1}\x_1 + \c_1]_\zero^{\m_1},
    % \end{align*}
    % is GES to an equilibrium.
    % \item
    the input $\c_1(t)$ lies in a compact set and has a bounded rate derivative;
  \item For each $i\in \{2,\dots,N\}\cup\{T\}$, the matrix $\W_{i,i}^{11}$ satisfies
    $\rho(|\W_{i,i}^{11}|) \leq \alpha_i$, with $\alpha_i < 1$;
  \item If
    $ \rho\big(\sum_{i=1}^N |\W_{i,T}^{11}\W_{T,i}^{11}|\big) \neq 0$,
    then $\alpha + \max(\delta,\delta^{1/p}) < 1$, where $p$ is the
    dimension of $\mathbf{N}_{\bar{\W}^{11}}$ and
    \begin{align*}
      \alpha
      & = \max_{i\in \{2,\dots,N\}\cup\{T\}} \{\alpha_i\}, \qquad 
        \delta = \gamma\sum_{j=1}^{p-1}
        \norm{\mathbf{N}_{\bar{\W}^{11}}}^j,
      \\
      \gamma
      & = \max \{\sum_{i=1}^{N-1}
        \W_{i,T}^{11}\W_{i,T}^{11^\top},\sum_{i=1}^{N-1}
        \W_{T,i}^{11}\W_{T,i}^{11^\top} \} .
    \end{align*}

  \end{enumerate}
  Then, there exist control laws
  $\u_i(t) = \K_i\x_i(t) + \bar{\u}_i(t)$, with
  $\K_i \in \R^{p_i \times n_i}$ and
  $\map{\bar{\u}_i}{\Rpluseq}{\Rpluseq^{p_i}}$, and $\epsilon > 0$
  such that the cortical and thalamic regions within the closed-loop
  system achieve $\epsilon-$selective inhibition and
  recruitment. Furthermore, if $\norm{\dot{\c}_1(t)} \to 0$ as
  $t \to \infty$, then the network achieves selective inhibition and
  recruitment~\eqref{eq:selective_inh_rec_star}.
  \end{theorem}
\begin{proof}
  First, for cortical region $\mathcal{N}_i$, $i \in \{2,\dots,N\}$,
  define the control laws
  $\u_i(t) = \B_i\K_i\x_i(t) + \B_i\bar{\u}(t)$ such that
  \begin{subequations}\label{eq:control_cortical_star}
    \begin{align} 
      \B_i\K_i &\leq -\W_{i,i}^{0,\all}, \label{eq:control_cortical_feedback_star} \\
      \B_i\bar{\u}_i(t) &\leq -\W_{i,T}^{0,\all}\x_T(t) -
                          \c_i^0. \label{eq:control_cortical_feedforward_star} 
    \end{align}
  \end{subequations}
  In a similar fashion, define the control law for the thalamus,
  $\mathcal{N}_T$ by $\u_T(t) = \B_T\K_T\x_T(t) + \B_T\bar{\u}(t)$
  such that it satisfies
  \begin{subequations}\label{eq:control_thalamus_star}
    \begin{align}
      \B_T\K_T &\leq -\W_T^{0,\all}, \label{eq:control_thalamus_feedback_star}\\
      \B_T\bar{\u}_T(t) &\leq -\W_{T,1}^{0,\all}\x_1(t) -
      \c_T^0. \label{eq:control_thalamus_feedforward_star}
    \end{align}
  \end{subequations}
  Now, we permute the system variables and define corresponding
  timescale matrices as follows
  \begin{align*}
    \bar{\x} &= \begin{bmatrix}
      \x^0 \\
      \x^1
    \end{bmatrix} \qquad
    \x^0 = \begin{bmatrix}
      \x_2^0 & \dots & \x_N^0 & \x_T^0 \end{bmatrix}^\top \\
    \x^1 &= \begin{bmatrix}
      \x_1 & \x_2^1 & \dots & \x_N^1 & \x_T^1
    \end{bmatrix}^\top \qquad
    \tau^0 = \text{diag}(\tau_2,\dots,\tau_N,\tau_T) \\
    \tau^1 &= \text{diag}(\tau_1,\tau_2,\dots,\tau_N,\tau_T).
  \end{align*}
  Substituting in the control laws~\eqref{eq:control_cortical_star}
  and~\eqref{eq:control_thalamus_star}, we have the following
  controlled system dynamics
  \begin{subequations}\label{eq:controlled_uniform_timescale_star}
  \begin{align}
    \tau^0\dot{\x}^0 &= -\x^0 \\
    \tau^1\dot{\x}^1 &= -\x^1 + [\bar{\W}\bar{\x} + \c(t)]_\zero^{\bar{\m}},
  \end{align}
  \end{subequations}
  where
  $\bar{\W} = \begin{bmatrix}\bar{\W}^{10} &
    \bar{\W}^{11} \end{bmatrix},$
  with $\bar{\W}^{11}$ as in~\eqref{eq:Wbar_matrix}, and $\c(t)$ the
  permutation of the signal $\c_1(t)$ and the constants $\c_i$,
  $i \in \{2,\dots,N\}\cup\{T\}$ corresponding to the permuted variables. Now, we
  consider a `frozen' version of the
  dynamics~\eqref{eq:controlled_uniform_timescale_star}, in which we
  fix $\c_1(t)$ to a constant $\bar{\c}_1$.  By~\cite[Lemma
  A.1]{EN-JC:21-tacI} the frozen version of the
  dynamics~\eqref{eq:controlled_uniform_timescale_star} is GES to an
  equilibrium $\x^*$, with $\x^0 \to \zero$, if
  $\dot{\bf y} = -{\bf y} + [\bar{\W}^{11}{\bf y} +
  \c]_\zero^{\bar{\m}}$ is GES to an equilibrium. By assumptions (ii)
  and (iii), along with~\cite[Theorem 2]{KEC:86}, we have that
  $\rho(\bar{\W}^{11}) < 1$ and therefore, the dynamics
  $\dot{\bf y} = -{\bf y} + [\bar{\W}^{11}{\bf y} +
  \c]_\zero^{\bar{\m}}$ is GES to an equilibrium, cf.~\cite[Theorem
  IV.8]{EN-JC:21-tacI}. Therefore the frozen version
  of~\eqref{eq:controlled_uniform_timescale_star} is GES to an
  equilibrium using the control laws defined above. Given this, we can
  apply the converse Lyapunov result given in
  Theorem~\ref{thrm:converse_lyap} to get a Lyapunov function for the
  dynamics~\eqref{eq:controlled_uniform_timescale_star} that satisfies
  the requirements of Theorem~\ref{thrm:slowly_varying_stability}. The
  direct application of this result then gives that selective
  inhibition and recruitment~\eqref{eq:selective_inh_rec_star} is
  achieved if $\norm{\dot{\c}_1(t)} \to 0$ and is
  $\epsilon$-selectively inhibited and recruited if it is bounded but
  does not tend to zero.
\end{proof}

Note that Theorem~\ref{thrm:star_uniform_timescale} relies on the
stability of the task-relevant dynamics of each network layer when
considered independently and the assumption that the magnitude of the
combination of thalamocortical and corticothalamic interconnections
does not exceed a certain stability margin. The latter condition is
consistent with neuroscientific observations: in fact, enhanced
corticothalamic feedback may result in pathological
behavior~\cite{KG-JMD:20}. In particular, excessive corticothalamic
feedback has been found to coincide with epileptic loss of
consciousness in absence seizures due to over-inhibition in the
cortical regions~\cite{GKK:01}.

\begin{remark}\longthmtitle{Remote synchronization in star-connected
    brain networks} 
  Remote synchronization is a phenomenon observed in the brain in
  which distant brain regions with similar structure synchronize their
  activity despite the lack of a direct link~\cite{VV-PH:14}. This
  should then naturally arise in star-connected networks if there is
  morphological symmetry between cortical regions, as this topology directly shows regions without direct
  links. From~\eqref{eq:equilibrium_sys_equations_star}, we note that
  if any two cortical regions $\mathcal{N}_i$ and $\mathcal{N}_j$ have
  identical task-relevant dynamics, i.e.,
  $\W_{i,i}^{1,\all} = \W_{j,j}^{1,\all}$,
  $\W_{i,T}^{1,\all} = \W_{j,T}^{1,\all}$ and $\c_i^1 = \c_j^1$, it
  follows that the equilibrium points will satisfy
  $\x_i^{1^*} = \x_j^{1^*}$, meaning that remote synchronization is
  achieved provided that the conditions of
  Theorem~\ref{thrm:star_uniform_timescale} are satisfied.

  Remarkably, this conclusion seems to be independent of the
  particular dynamics of the individual layers. In fact, the
  work~\cite{YQ-YK-MC:18-cdc} studies remote synchronization in star-connected brain networks, cf.
  Fig.~\ref{fig:thalamocortical-networks}(b), with layer dynamics
  given by Kuramoto oscillator dynamics,
  \begin{align}\label{eq:kuramoto_dynamics}
    \dot{\theta}_T &= \omega_0 + \sum_{i=1}^{N} K_i\sin(\theta_i -
    \theta_0 - \xi) \notag
    \\
    \dot{\theta}_i &= \omega + A_i\sin(\theta_0-\theta_i-\xi),
    \qquad i = 1,\dots,N.
  \end{align}
  According to~\cite{YQ-YK-MC:18-cdc}, the outer cortical regions in
  star-connected brain networks can remotely synchronize despite no
  direct links between the regions provided the network dynamics
  satisfy conditions that parallel those required for the
  star-connected linear-threshold networks studied above. In
  particular, to be able to achieve remote synchronization, the
  network weights must satisfy $A_i \geq (N-1)K_i$ for all
  $i \in \until{N}$. This condition guarantees the existence of a
  locally stable equilibrium point, and is equivalent the requirement
  of the matrices defining the task-relevant components of the
  linear-threshold network being individually stable.  \oprocend
\end{remark}

We conclude by noting that the thalamus, as a relay, can function as a
failsafe for the hierarchical thalamocortical network, allowing for
selective inhibition and recruitment even if corticocortical
connections become damaged.  In fact, observe that the hierarchical
topology where the matrices $\W_{i,i-1},\W_{i,i+1}$ are equal to zero
for all $i$ reduces to the star-connected topology, while maintaining
the timescale separation between layers. Therefore, the star-connected
topology with a hierarchical timescale structure can be considered as
a failsafe for the hierarchical thalamocortical network.  The
next result provides conditions for selective inhibition and
recruitment for this topology.

\begin{corollary}\longthmtitle{Selective inhibition of a
    star-connected hierarchical thalamocortical network}\label{cor:hierarchical_star_connected}
  Consider a hierarchical star-connected thalamocortical network of
  the form shown in Fig.~\ref{fig:thalamocortical-networks}(b) with
  timescales $\tau_1 \gg \tau_2 \gg \dots \gg \tau_N$ and layer
  dynamics given by~\eqref{eq:cortical_layer_dynamics}
  and~\eqref{eq:thalamus_layer_dynamics}. Without loss of generality
  let $\tau_T$ be such that $\tau_1 \gg \tau_T \gg \tau_N$ and let
  $a,b \in \{1,\dots,N\}$ such that $\tau_a \gg \tau_T \gg \tau_b$ and
  $b = a+1$. Assume the stability
  assumptions~\eqref{eq:reduced_dynamics_top}-\eqref{eq:reduced_dynamics_below_thal}
  for the reduced-order subnetworks are satisfied. Then, for
  $i \in \{1,\dots,N\}\cup\{T\}$ and constants $\c_i \in \R^{n_i}$ and
  $\c_T \in \R^{n_T}$, there exist control laws
  $\u_i(t) = \K_i\x_i(t) + \bar{\u_i}(t)$, with
  $\K_i \in \R^{p_i \times n_i}$ and
  $\map{\bar{\u}_i}{\Rpluseq}{\Rpluseq^{p_i}}$, such that the
  closed-loop system achieves selective inhibition and
  recruitment~\eqref{eq:selective-recruitment-inhibition}.
\end{corollary}

The proof of the result is similar to that of
Theorem~\ref{thrm:multilayer_LTR}, with differences occurring in the
constructed control laws on the basis that
$\W_{i,i+1} = \W_{i,i-1} = \zero$ for all $i \in \until{N}$.  The loss
of these connections plays a significant role in the form of the
control. In particular, the amount of feedforward control coming from
the thalamus to the cortical regions increases, due to the fact that
direct feedforward control between cortical regions is not possible.
Fig.~\ref{fig:sample_trajectories-star} illustrates
Corollary~\ref{cor:hierarchical_star_connected} on the star-connected
network obtained by removing the direct connections between cortical
regions in the network of Fig.~\ref{fig:sample_trajectories}.

\begin{figure}[htb]
    \centering
  \includegraphics[scale=0.18]{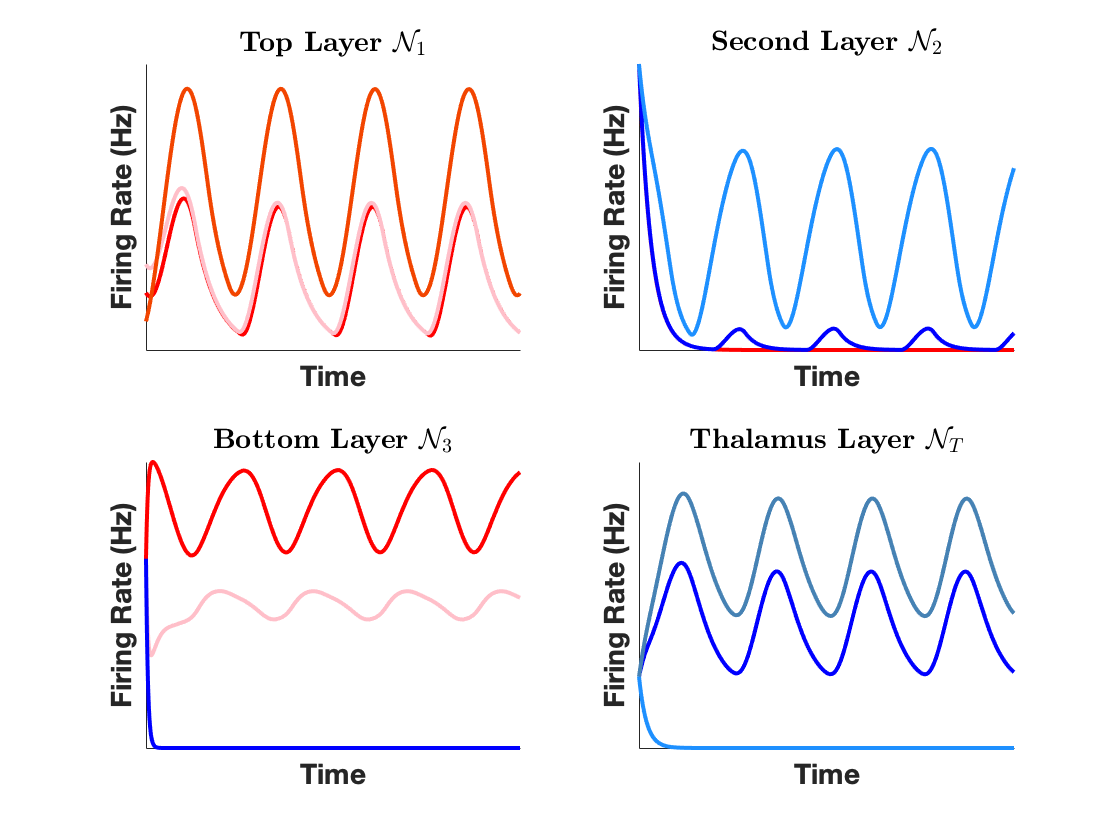}
  \caption{Trajectories of the network in
    Fig.~\ref{fig:sample_trajectories} with the connections between
    the cortical regions removed. The star-connected topology allows
    for successful selective inhibition of the desired set of nodes
    despite the lack of connections between cortical regions (cf.
 Corollary~\ref{cor:hierarchical_star_connected}),
    providing a failsafe topology.
  }\label{fig:sample_trajectories-star}
    \vspace*{-1ex}
\end{figure}

\section{Quantitative Comparison of Cortical and Thalamocortical
  Networks}\label{sec:examples}

This section seeks to quantitatively illustrate ways in which the
presence of the thalamus might have a beneficial effect in the
behavior and performance of the dynamic models for brain networks
adopted here.  Fig.~\ref{fig:sample_trajectories-star} has already
illustrated the failsafe role played by the thalamus in hierarchical
thalamocortical networks.  Here we focus on two other beneficial
impacts of the thalamus that we have observed in simulation: the
required control magnitude to achieve selective inhibition and the
convergence time in thalamocortical networks versus cortical ones.

\begin{example}\longthmtitle{Reduced
    average control magnitude in thalamocortical vs cortical networks} 
  We investigate the control magnitude required to achieve selective
  inhibition.  Control magnitude here refers to the aggregate of the
  inputs at all layers integrated over time and averaged across
  trials.  We consider hierarchical pairs of thalamocortical and
  cortical networks, where the latter is obtained by disconnecting the
  thalamus in the former.  Fig.~\ref{fig:average_control} shows that
  thalamocortical networks require a lower control magnitude to
  achieve selective inhibition in the cortical regions relative to the
  corresponding strictly cortical networks, matching the intuition
  that they are easier to selectively inhibit due to the thalamus
  impacting the cortical regions in an inhibitory fashion.

\oprocend

\end{example}

\begin{figure}[htb]
  \centering
  \includegraphics[scale=0.2]{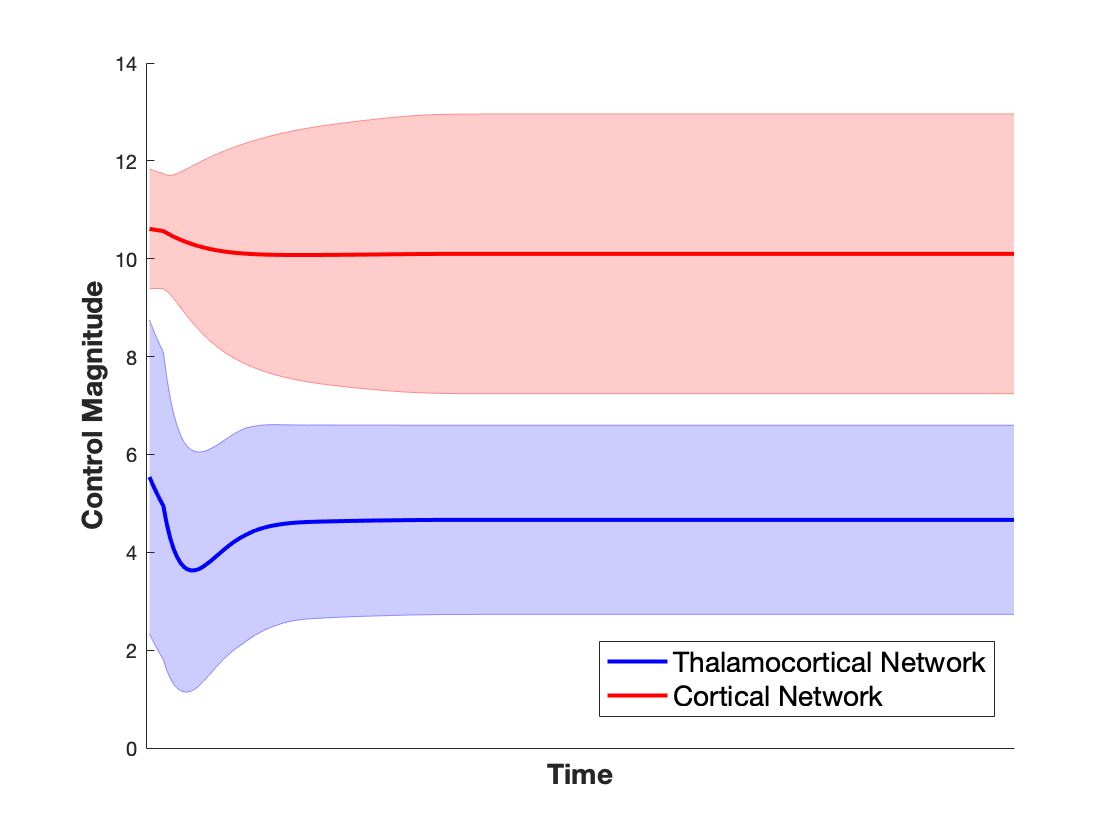}
  \caption{Comparison of average control magnitude in cortical and
    thalamocortical networks. The plot displays the control magnitude
    required to selectively inhibit the bottom cortical layer averaged
    over $100$ random thalamocortical networks and the corresponding
    cortical network in which the thalamus is removed. Shaded regions
    correspond to the error bars.  All networks are composed of two
    three-node cortical regions. The top cortical layer has two
    excitatory nodes and the bottom layer has only one. The thalamus
    is composed of two inhibitory nodes. Each thalamocortical and
    cortical network generated satisfies the assumptions of
    Theorem~\ref{thrm:multilayer_LTR}
    and~\cite[Theorem~IV.3]{EN-JC:21-tacII}, resp., along with
    biological sign constraints.  To make the required control
    magnitude directly comparable, the thalamocortical and cortical
    networks are inhibiting the same set of nodes within the bottom
    cortical layer. }\label{fig:average_control}
    \vspace*{-1ex}
\end{figure}

\begin{example}\longthmtitle{Convergence time of thalamocortical and
    cortical networks}
  Here we consider the speed at which thalamocortical networks
  converge to an equilibrium as another metric to evaluate the role of
  the thalamus. We compare the convergence time for a cortical network
  with that of a thalamocortical network maintaining the same cortical
  regions. While performing the comparison is interesting as a
  function of multiple network parameters (e.g., network size, layer
  size, ratio of excitatory-inhibitory nodes), we focus here
  specifically on the thalamus and, in particular, on varying its
  timescale with respect to the cortical regions in the
  model. Thalamocortical networks with varying timescales are of
  particular interest due to their existence in the brain, as even
  restricting only to the visual thalamus, the thalamus operates at
  both slow and fast timescales~\cite{ZY-XY-CMH-ZA-NPF-WW-SGB:17}.
  Fig.~\ref{fig:convergence_speed_timescales} shows that
  thalamocortical networks have faster average convergence time, with
  the margin between the two networks decreasing as the timescale
  $\tau_T$ increases.  This validates the the beneficial role played
  by the thalamus, with faster thalamus dynamics (smaller $\tau_T$)
  helping the cortical regions converge faster, leading to  overall
  decreased  convergence time.
  \oprocend
\end{example}

\begin{figure}
  \centering
  \includegraphics[scale=0.19]{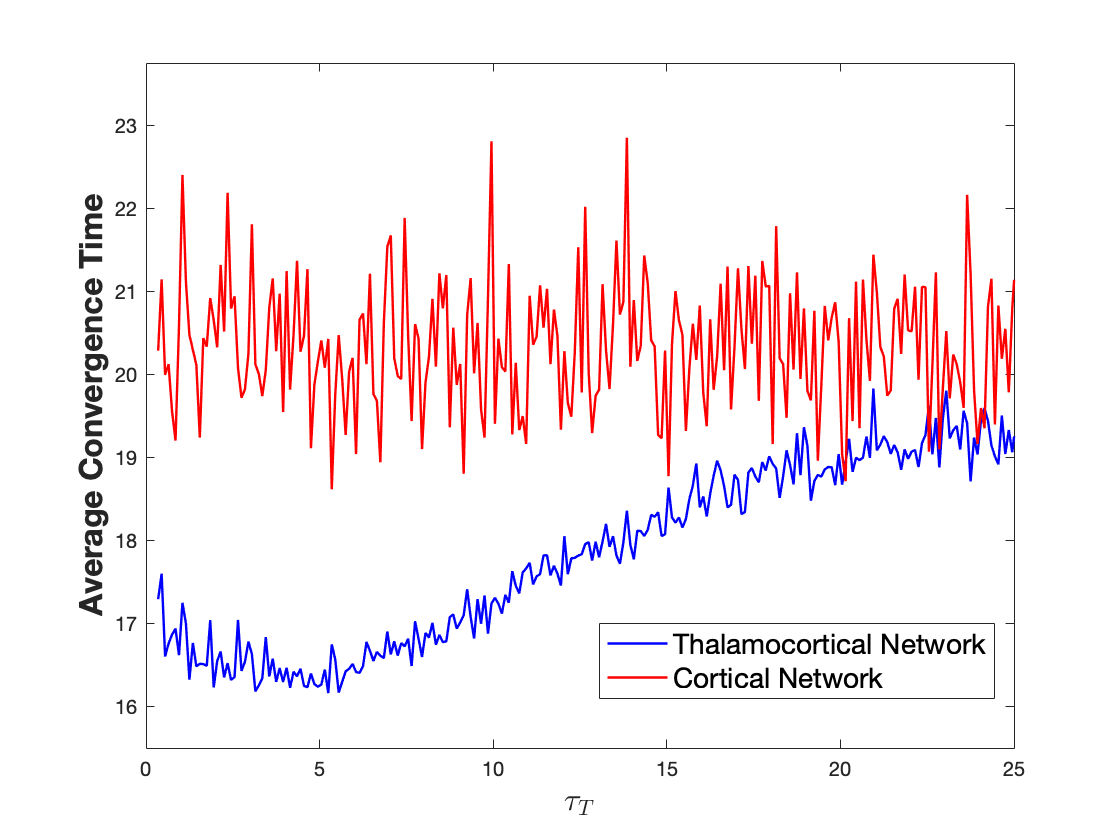}
  \caption{Convergence time of thalamocortical network with changing
    timescales. The graph illustrates the average convergence time to
    the equilibrium for $30$ randomly generated cortical and
    thalamocortical networks with $20$ layers, each containing two
    excitatory and two inhibitory nodes. For each simulation the
    initial condition is set to be a uniform distance away from the
    equilibrium and the set of nodes to be inhibited is randomly
    selected.  The average convergence time for nodes within the
    thalamocortical network is lower than for the cortical network. As
    the timescale $\tau_T$ of the thalamus increases, the margin of
    improvement decreases. }\label{fig:convergence_speed_timescales}
  \vspace*{-1ex}
\end{figure}

\section{Conclusions}
We have studied the properties of both multilayer hierarchical and
star-connected thalamocortical brain networks modeled with
linear-threshold dynamics. Our primary motivation was understanding
the role played by the thalamus in achieving selective inhibition and
recruitment of neural populations.  For both types of interconnection
topologies, we have described how the equilibria at each layer depends
on the rest of the network and identified suitable stability
conditions.  For hierarchical networks, these take the form of GES
requirements of the reduced-order dynamics of individual layers.  For
star-connected thalamocortical network without a hierarchy of
timescales, these take the form of stability of the task-relevant
dynamics of each layer when considered independently and the magnitude
of the combination of thalamocortical and corticothalamic
interconnections not exceeding a certain stability margin.  Future work
will seek to analytically characterize the robustness and performance
of thalamocortical networks, study the role of the thalamus in other
cognitive tasks beyond selective attention (e.g., sleep and
consciousness, oscillations, and learning), and explore the impact of
the addition and deletion of neuronal populations (neurogenesis)
in the performance and expressivity of brain networks.

{
\bibliographystyle{ieeetr}

\begin{thebibliography}{10}

\bibitem{SMS:12}
S.~M. Sherman, ``Thalamocortical interactions,'' {\em Current Opinion in
  Neurobiology}, vol.~22, no.~4, pp.~575--579, 2012.

\bibitem{NT:50}
N.~Tinbergen, ``The hierarchical organization of nervous mechanisms underlying
  instinctive behaviour,'' in {\em Symposium for the Society for Experimental
  Biology}, vol.~4, pp.~305--312, 1950.

\bibitem{ARL:70}
A.~R. Luria, ``The functional organization of the brain,'' {\em Scientific
  American}, vol.~222, no.~3, pp.~66--79, 1970.

\bibitem{SJK-JD-KJF:08}
S.~J. Kiebel, J.~Daunizeau, and K.~J. Friston, ``A hierarchy of time-scales and
  the brain,'' {\em PLOS Computational Biology}, vol.~4, no.~11, p.~e1000209,
  2008.

\bibitem{JDM-AB-DJF-RR-JDW-XC-CP-TP-HS-DL-XW:14}
J.~D. Murray, A.~Bernacchia, D.~J. Freedman, R.~Romo, J.~D. Wallis, X.~Cai,
  C.~Padoa-Schioppa, T.~Pasternak, H.~Seo, D.~Lee, and X.~Wang, ``A hierarchy
  of intrinsic timescales across primate cortex,'' {\em Nature Neuroscience},
  vol.~17, no.~12, p.~1661, 2014.

\bibitem{DJF-DCEV:91}
D.~J. Felleman and D.~C.~V. Essen, ``Distributed hierarchical processing in the
  primate cerebral cortex,'' {\em Cerebral Cortex}, vol.~1, no.~1, pp.~1--47,
  1991.

\bibitem{UH-JC-CJH:15}
U.~Hasson, J.~Chen, and C.~J. Honey, ``Hierarchical process memory: memory as
  an integral component of information processing,'' {\em Trends in Cognitive
  Sciences}, vol.~19, no.~6, pp.~304--313, 2015.

\bibitem{MIR-IT-PV:15}
M.~I. Rabinovich, I.~Tristan, and P.~Varona, ``Hierarchical nonlinear dynamics
  of human attention,'' {\em Neuroscience \& Biobehavioral Reviews}, vol.~55,
  pp.~18--35, 2015.

\bibitem{KH-MAB-WBL-MD:17}
K.~Hwang, M.~A. Bertolero, W.~B. Liu, and M.~D'Esposito, ``The human thalamus
  is an integrative hub for functional brain networks,'' {\em Journal of
  Neuroscience}, vol.~37, no.~23, pp.~5594--5607, 2017.

\bibitem{RDD-AB:17}
R.~D. D’Souza and A.~Burkhalter, ``A laminar organization for selective
  cortico-cortical communication,'' {\em Frontiers in Neuroanatomy}, vol.~11,
  p.~71, 2017.

\bibitem{MMH-LA-16}
M.~M. Halassa and L.~Acs{\'a}dy, ``Thalamic inhibition: diverse sources,
  diverse scales,'' {\em Trends in Neurosciences}, vol.~39, no.~10,
  pp.~680--693, 2016.

\bibitem{JMA-HAS:15}
J.~M. Alonso and H.~A. Swadlow, ``Thalamus controls recurrent cortical
  dynamics,'' {\em Nature Neuroscience}, vol.~18, no.~12, pp.~1703--1704, 2015.

\bibitem{SMS-RWG:06}
S.~M. Sherman and R.~W. Guillery, {\em Exploring the Thalamus and Its Role in
  Cortical Function}.
\newblock MIT press, 2006.

\bibitem{LG-SPJ-DEF-MC-MS:05}
L.~Gabernet, S.~P. Jadhav, D.~E. Feldman, M.~Carandini, and M.~Scanziani,
  ``Somatosensory integration controlled by dynamic thalamocortical
  feed-forward inhibition,'' {\em Neuron}, vol.~48, no.~2, pp.~315--327, 2005.

\bibitem{SJC-TJL-BWC:07}
S.~Cruikshank, T.~J. Lewis, and B.~Connors, ``Synaptic basis for intense
  thalamocortical activation of feedforward inhibitory cells in neocortex,''
  {\em Nature Neuroscience}, vol.~10, no.~4, pp.~462--468, 2007.

\bibitem{JAH-SM-KEH-JDW-HC-AB-PB-SC-LC-AC:19}
J.~A. Harris, S.~Mihalas, K.~E. Hirokawa, J.~D. Whitesell, H.~Choi, A.~Bernard,
  P.~Bohn, S.~Caldejon, L.~Casal, A.~Cho, {\em et~al.}, ``Hierarchical
  organization of cortical and thalamic connectivity,'' {\em Nature}, vol.~575,
  no.~7781, pp.~195--202, 2019.

\bibitem{RC-KK-MG-HK-XW:15}
R.~Chaudhuri, K.~Knoblauch, M.~Gariel, H.~Kennedy, and X.~Wang, ``A large-scale
  circuit mechanism for hierarchical dynamical processing in the primate
  cortex,'' {\em Neuron}, vol.~88, no.~2, pp.~419--431, 2015.

\bibitem{KM-AD-VI-CC:16}
K.~Morrison, A.~Degeratu, V.~Itskov, and C.~Curto, ``Diversity of emergent
  dynamics in competitive threshold-linear networks: a preliminary report,''
  {\em arXiv preprint arXiv:1605.04463}, 2016.

\bibitem{EN-JC:21-tacI}
E.~Nozari and J.~Cort\'es, ``Hierarchical selective recruitment in
  linear-threshold brain networks. {P}art {I}: Intra-layer dynamics and
  selective inhibition,'' {\em IEEE Transactions on Automatic Control},
  vol.~66, no.~3, pp.~949--964, 2021.

\bibitem{EN-JC:21-tacII}
E.~Nozari and J.~Cort\'es, ``Hierarchical selective recruitment in
  linear-threshold brain networks. {P}art {II}: Inter-layer dynamics and
  top-down recruitment,'' {\em IEEE Transactions on Automatic Control},
  vol.~66, no.~3, pp.~965--980, 2021.

\bibitem{DL:03}
D.~Liberzon, {\em Switching in Systems and Control}.
\newblock Systems \& Control: Foundations \& Applications, Birkh{\"a}user,
  2003.

\bibitem{MKJJ:03}
M.~K.~J. Johansson, {\em Piecewise Linear Control Systems: A Computational
  Approach}.
\newblock Lecture Notes in Control and Information Sciences, Springer Berlin
  Heidelberg, 2003.

\bibitem{ANT:52}
A.~N. Tikhonov, ``Systems of differential equations containing small parameters
  in the derivatives,'' {\em Matematicheskii Sbornik}, vol.~73, no.~3,
  pp.~575--586, 1952.

\bibitem{PVK-HKK:99}
P.~V. Kokotovi\'c and H.~K. Khalil, eds., {\em Singular Perturbation Methods in
  Control: Analysis and Design}.
\newblock SIAM, 1999.

\bibitem{VV:97}
V.~Veliov, ``A generalization of the {T}ikhonov theorem for singularly
  perturbed differential inclusions,'' {\em Journal of Dynamical \& Control
  Systems}, vol.~3, no.~3, pp.~291--319, 1997.

\bibitem{PD-LFA:01}
P.~Dayan and L.~F. Abbott, {\em Theoretical Neuroscience: Computational and
  Mathematical Modeling of Neural Systems}.
\newblock Computational Neuroscience, Cambridge, MA: MIT Press, 2001.

\bibitem{JC-UH-CJH:15}
J.~Chen, U.~Hasson, and C.~Honey, ``Processing timescales as an organizing
  principle for primate cortex,'' {\em Neuron}, vol.~88, no.~2, pp.~244--246,
  2015.

\bibitem{SMS-RWG:11}
S.~M. Sherman and R.~W. Guillery, ``Distinct functions for direct and
  transthalamic corticocortical connections,'' {\em Journal of
  Neurophysiology}, vol.~106, no.~3, pp.~1068--1077, 2011.

\bibitem{KR-ADL-MS:15}
K.~Reinhold, A.~D. Lien, and M.~Scanziani, ``Distinct recurrent versus afferent
  dynamics in cortical visual processing,'' {\em Nature Neuroscience}, vol.~18,
  no.~12, pp.~1789--1797, 2015.

\bibitem{ASM-SMS-MAS-RGM-RPV-YC:14}
A.~S. Mitchell, S.~M. Sherman, M.~A. Sommer, R.~G. Mair, R.~P. Vertes, and
  Y.~Chudasama, ``Advances in understanding mechanisms of thalamic relays in
  cognition and behavior,'' {\em Journal of Neuroscience}, vol.~34, no.~46,
  pp.~15340--15346, 2014.

\bibitem{FA-VF-ARM-EJK-EC-WM:18}
F.~Alcaraz, V.~Fresno, A.~R. Marchand, E.~J. Kremer, E.~Coutureau, and
  M.~Wolff, ``Thalamocortical and corticothalamic pathways differentially
  contribute to goal-directed behaviors in the rat,'' {\em Elife}, vol.~7,
  p.~e32517, 2018.

\bibitem{JSI-MS:11}
J.~S. Isaacson and M.~Scanziani, ``How inhibition shapes cortical activity,''
  {\em Neuron}, vol.~72, no.~2, pp.~231--243, 2011.

\bibitem{KEC:86}
K.~E. Chu, ``Generalization of the bauer-fike theorem,'' {\em Numerische
  Mathematik}, vol.~49, no.~6, pp.~685--691, 1986.

\bibitem{KG-JMD:20}
K.~George and J.~M. Das, ``Neuroanatomy, thalamocortical radiations,'' {\em
  StatPearls [Internet]}, 2020.

\bibitem{GKK:01}
G.~K. Kostopoulos, ``Involvement of the thalamocortical system in epileptic
  loss of consciousness,'' {\em Epilepsia}, vol.~42, pp.~13--19, 2001.

\bibitem{VV-PH:14}
V.~Vuksanovi{\'c} and P.~H{\"o}vel, ``Functional connectivity of distant
  cortical regions: role of remote synchronization and symmetry in
  interactions,'' {\em NeuroImage}, vol.~97, pp.~1--8, 2014.

\bibitem{YQ-YK-MC:18-cdc}
Y.~Qin, Y.~Kawano, and M.~Cao, ``Stability of remote synchronization in star
  networks of kuramoto oscillators,'' in {\em {IEEE} Conf.\ on Decision and
  Control}, (Miami Beach, USA), pp.~5209--5214, 2018.

\bibitem{ZY-XY-CMH-ZA-NPF-WW-SGB:17}
Z.~Ye, X.~Yu, C.~M. Houston, Z.~Aboukhalil, N.~P. Franks, W.~Wisden, and S.~G.
  Brickley, ``Fast and slow inhibition in the visual thalamus is influenced by
  allocating {$\text{GABA}_A$} receptors with different $\gamma$ subunits,''
  {\em Frontiers in Cellular Neuroscience}, vol.~11, p.~95, 2017.

\bibitem{SL-HA:15}
S.~Lee and H.~Ahn, ``Robust stability of slowly varying nonlinear systems
  having a continuum of equilibria,'' in {\em {IEEE} Conf.\ on Decision and
  Control}, (Osaka, Japan), IEEE, Dec. 2015.

\bibitem{HKK:02}
H.~K. Khalil, {\em Nonlinear Systems}.
\newblock Prentice Hall, 3~ed., 2002.

\end{thebibliography}

}

%% BIOS FOR FINAL VERSION
% \vspace{-6ex}

\begin{IEEEbiography}[{\includegraphics[width=1in,height=1.25in,clip,keepaspectratio]{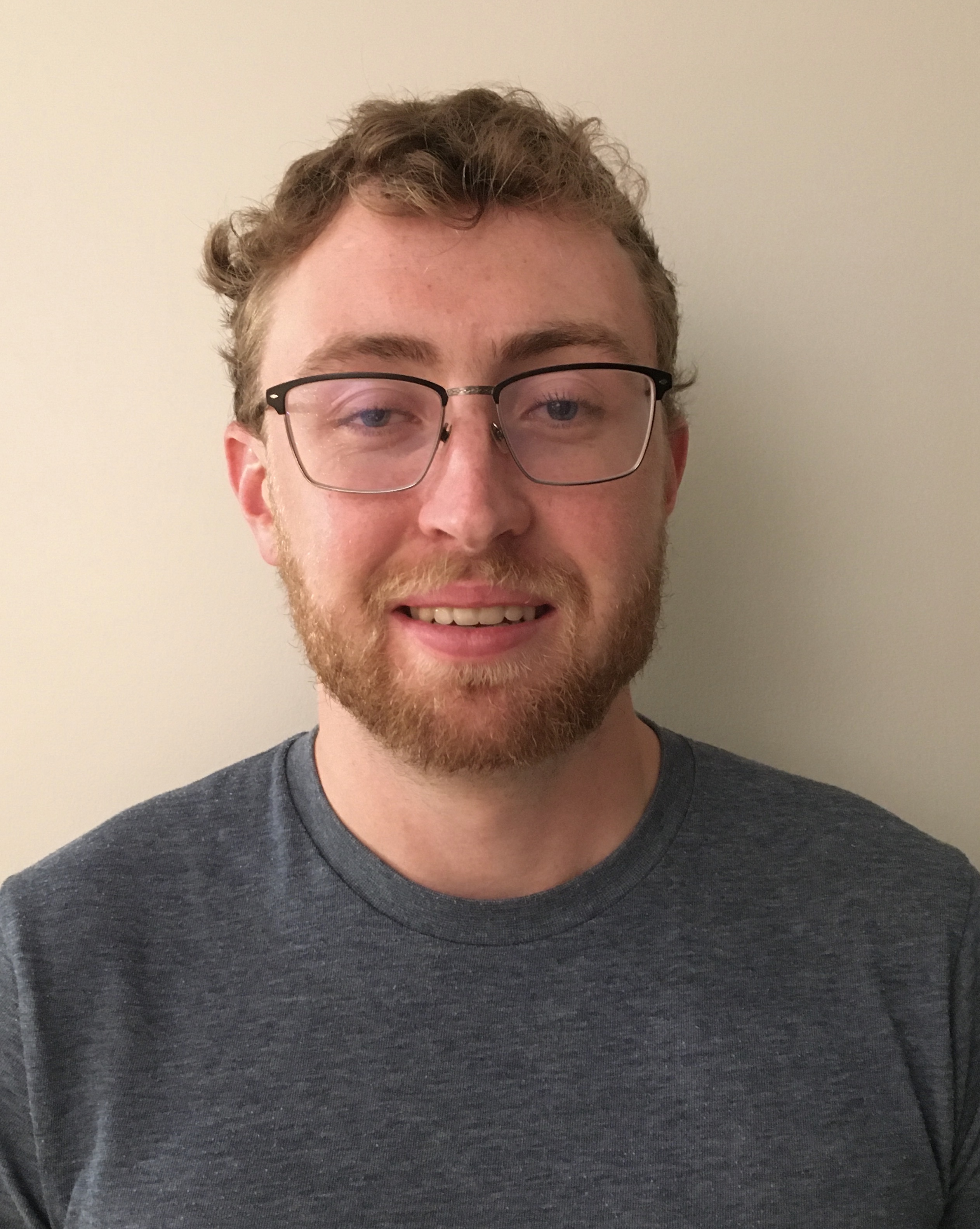}}]{Michael
    McCreesh} received his B.A.Sc degree in Mathematics and
  Engineering and his M.A.Sc in Mathematics and Engineering from
  Queen's University, Kingston, Canada in 2017 and 2019,
  resp. He is currently a Ph.D. student in the Department of
  Mechanical and Aerospace Engineering at the University of California
  San Diego. His current research interests include control theory and
  its application to theoretical neuroscience, in particular the
  application of dynamical systems to model brain networks.
\end{IEEEbiography}

\vspace*{-6ex}

\begin{IEEEbiography}[{\includegraphics[width=1in,height=1.25in,clip,keepaspectratio]{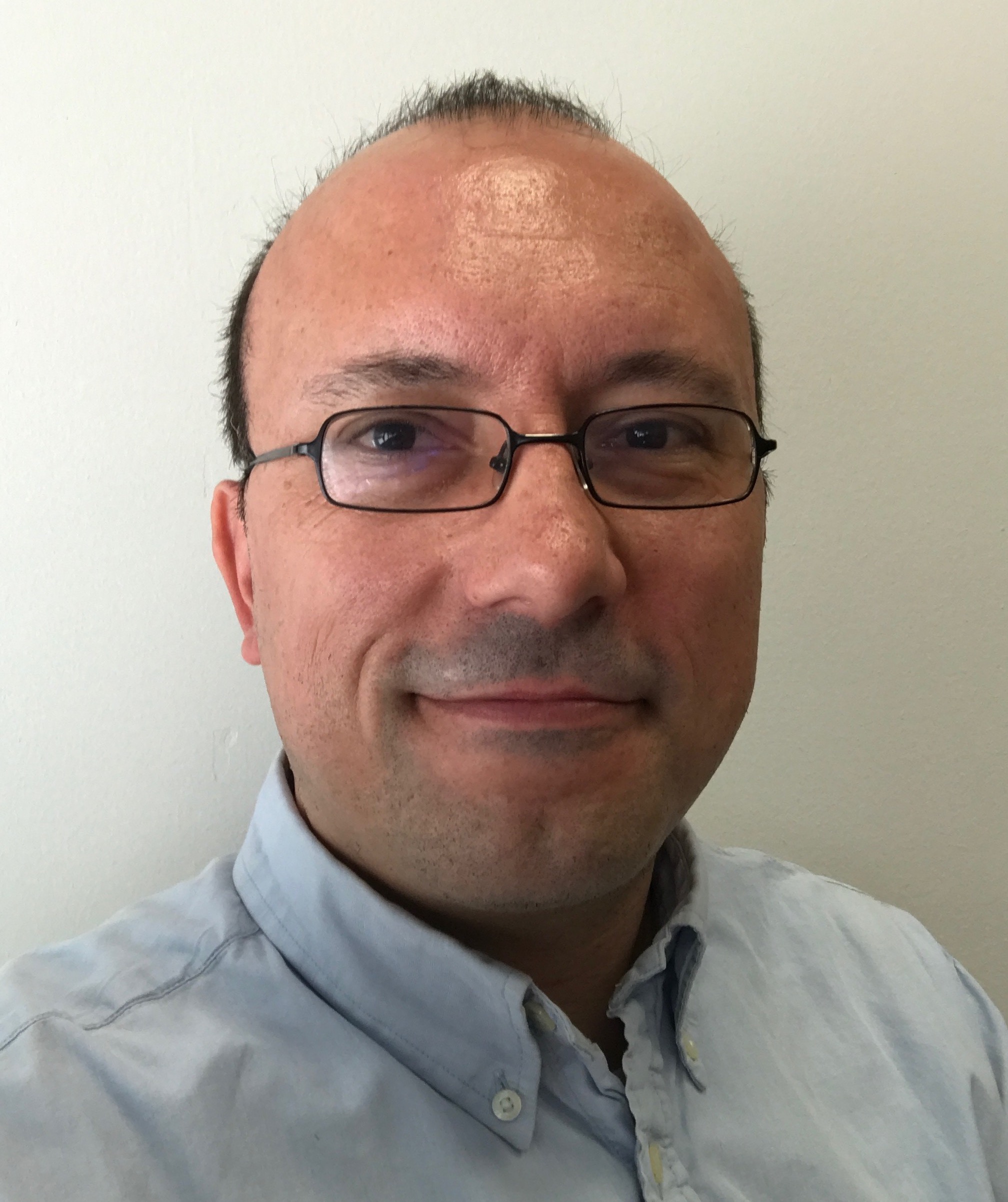}}]{Jorge
    Cort\'{e}s}
  (M'02, SM'06, F'14) received the Licenciatura degree in mathematics
  from Universidad de Zaragoza, Zaragoza, Spain, in 1997, and the
  Ph.D. degree in engineering mathematics from Universidad Carlos III
  de Madrid, Madrid, Spain, in 2001. He held postdoctoral positions
  with the University of Twente, Twente, The Netherlands, and the
  University of Illinois at Urbana-Champaign, Urbana, IL, USA. He was
  an Assistant Professor with the Department of Applied Mathematics
  and Statistics, University of California, Santa Cruz, CA, USA, from
  2004 to 2007. He is currently a Professor in the Department of
  Mechanical and Aerospace Engineering, University of California, San
  Diego, CA, USA.  He is the author of Geometric, Control and
  Numerical Aspects of Nonholonomic Systems (Springer-Verlag, 2002)
  and co-author (together with F. Bullo and S.  Mart{\'\i}nez) of
  Distributed Control of Robotic Networks (Princeton University Press,
  2009).  He is a Fellow of IEEE and SIAM.
  % At the IEEE Control Systems Society, he has been a Distinguished
  % Lecturer (2010-2014), and is currently its Director of Operations
  % and an elected member (2018-2020) of its Board of Governors.
  His current research interests include distributed control and
  optimization, network science, nonsmooth analysis, reasoning and
  decision making under uncertainty, network neuroscience, and
  multi-agent coordination in robotic, power, and transportation
  networks.
\end{IEEEbiography}

\appendix

For completeness, here we include two results on the stability of
slowly varying nonlinear systems to a continuum equilibria,
generalized from~\cite[Theorem 3.1]{SL-HA:15} to the case of
exponential stability. Let
\begin{align}\label{eq:nonlin_system}
  \dot{\x} = f(\x(t),\d(t)),
\end{align}
where $\x(t) \in \R^n$ and $\d(t) \in \mathcal{D}$, and $\mathcal{D}$
is a compact subset of $\R^m$. We assume $f$ is continuous on
$\R^n \times \mathcal{D}$ and is locally Lipschitz in both $\x$ and
$\d$. We further consider the `frozen' version of the
system~\eqref{eq:nonlin_system} with a fixed parameter $\d$,
\begin{align}\label{eq:frozen_sys}
  \dot{\x} = f(\x(t),\d), \quad \d \in \mathcal{D}.
\end{align}
We denote the solution to~\eqref{eq:frozen_sys} for each initial
condition $\x(t_0) = \x_0$ and $\d \in \mathcal{D}$ by
$\x(t,\x_0,\d)$. Let $\mathcal{A}$ be a forward invariant set of the
system~\eqref{eq:frozen_sys}.

\begin{theorem}\longthmtitle{Exponential Stability of Slowly Varying
    Nonlinear Systems}\label{thrm:slowly_varying_stability}
  Consider the nonlinear system~\eqref{eq:nonlin_system} and assume
  there exists a continuously differentiable function
  $\map{V}{\R^n \times \mathcal{D}}{\R}$ such that
  \begin{subequations}
    \begin{align}
      k_1|\x|_\A^2 &\leq V(\x,\d) \leq
                     k_2|\x|_\A^2 \label{eq:slowly_varying_lyap_req1}
      \\
      \frac{\partial V(\x,\d)}{\partial \x}f(\x,\d) & \leq
                                                      -k_3|\x|_\A^2 \label{eq:slowly_varying_lyap_req2}
      \\
      \left|\frac{\partial V(\x,\d)}{\partial \d} \right| &\leq
                                                            k, \label{eq:slowly_varying_lyap_req3} 
    \end{align}
  \end{subequations}
  for all $\x \in \R^n$ and $\d \in \mathcal{D}$, with $k_1,k_2,k_3$
  and $k$ nonnegative constants.  If $\norm{\dot{d}(t)}$ is uniformly
  bounded in time, then there exist constants $\gamma,\lambda$ and $T$
  such that
  \begin{subequations}
    \begin{align}
      |\x(t,\x_0,\d)|_\A &\leq \gamma\norm{\x_0}e^{-\lambda(t-t_0)}
                           \quad \forall~ t_0 \leq t \leq t_0+T
      \\
      |\x(t,\x_0,\d)|_\A & \leq \frac{k_3}{k_2} |\x_0|_\A \quad \forall~ t \geq t_0,
    \end{align}
    If $\lim_{t\to \infty} \norm{\dot{\d}(t)} = 0$, then the system is
    exponentially stable.
  \end{subequations}
\end{theorem}

The proof follows a similar line of reasoning as in~\cite[Theorem
3.1]{SL-HA:15} with slight modifications in the use of stability
results from~\cite{HKK:02} to account for exponential stability. We
now provide a converse Lyapunov result, modified from~\cite[Theorem
3.2]{SL-HA:15}, complementary to the above result.

\begin{theorem}\longthmtitle{Converse Lyapunov Theorem for Slowly
    Varying Nonlinear Systems}\label{thrm:converse_lyap}
  If, for every fixed $\d \in \mathcal{D}$, the nonlinear
  system~\eqref{eq:nonlin_system} is GES to an equilibrium, then there
  exists a Lyapunov function that
  satisfies~\eqref{eq:slowly_varying_lyap_req1},~\eqref{eq:slowly_varying_lyap_req2}
  and~\eqref{eq:slowly_varying_lyap_req3}.
\end{theorem}

The proof of this result is identical to~\cite[Theorem 3.2]{SL-HA:15},
with the assumption of GES giving the desired function form in the
final step.

\end{document}